\documentclass[ejs, preprint]{imsart}
\usepackage{graphicx,verbatim,array,multicol,courier}
\usepackage{amsmath,amssymb,amsthm,mathrsfs}
\usepackage[pdftex,bookmarks,colorlinks]{hyperref}
\usepackage{topcapt}
\usepackage{chicago}
\usepackage[dvipsnames]{xcolor}
\usepackage[section]{placeins} 
\usepackage{mathdots}
\usepackage{mathtools} 
\usepackage{booktabs} 
\usepackage{dsfont}

\startlocaldefs

\def\thick#1{\hbox{\rlap{$#1$}\kern0.25pt\rlap{$#1$}\kern0.25pt$#1$}}

\def\spa{\mathcal{S}}
\def\spaint{\mathring{\spa}}
\def\Hint{\mathring{H}}

\def\bbeta{{\boldsymbol\beta}}

\def\bdeta{{\boldsymbol\eta}}
\def\btheta{{\boldsymbol\theta}}

\def\bTheta{{\boldsymbol\Theta}}

\def\indic{\mathds{1}}

\def\real{{\mathbb R}}

\def\prob{{\mathbb P}}

\def\nat{{\mathbb N}}

\def\lik{\mathcal{L}}

\def\spA{\mathcal{A}}

\def\spH{\mathcal{H}}
\def\spP{\mathcal{P}}

\def\spE{\mathscr{E}}
\def\setB{\mathscr{B}}

\newcommand{\diff}{\mathrm{d}}
\newcommand{\betaf}{\mathrm{Be}}

\newcommand{\by}{{\boldsymbol{y}}}

\newcommand{\bY}{{\boldsymbol{Y}}}
\newcommand{\bZ}{{\boldsymbol{Z}}}

\numberwithin{equation}{section}

\theoremstyle{definition}
\newtheorem{defi}{Definition}[section]

\theoremstyle{plain}

\newtheorem{prop}[defi]{Proposition}

\newtheorem{cor}[defi]{Corollary}
\newtheorem{algo}[defi]{Algorithm}
\newtheorem{assum}[defi]{Assumption}

\theoremstyle{remark}

\endlocaldefs

\doi{10.1214/16-EJS1162}
\pubyear{0000}
\volume{0}
\firstpage{0}
\lastpage{0}

\begin{document}

\begin{frontmatter}

\title{Bayesian Inference for the Extremal Dependence}
\runtitle{Bayesian Inference for the Extremal Dependence}

\begin{aug}

\author{\fnms{Giulia} \snm{Marcon}\ead[label=e1]{giulia.marcon@unipv.it}}
\address{Department of Political and Social Sciences, University of Pavia,\\ Corso Strada Nuova 65, 27100, Italy \\\printead{e1}}
\author{\fnms{Simone A.} \snm{Padoan}\corref{}\ead[label=e2]{simone.padoan@unibocconi.it}}
\and
\author{\fnms{Isadora} \snm{Antoniano-Villalobos}\ead[label=e3]{isadora.antoniano@unibocconi.it}}
\address{Department of Decision Sciences, Bocconi University, Milan,\\ via Roentgen 1, 20136, Italy \\ \printead{e2,e3}}

\runauthor{G. Marcon et al.}

\end{aug}

\begin{abstract}
A simple approach for modeling multivariate extremes is to consider
the vector of component-wise maxima and their max-stable distributions.
The extremal dependence can be inferred by estimating
the angular measure or, alternatively, 
the Pickands dependence function. We propose a nonparametric Bayesian model that allows, in the bivariate case, the simultaneous estimation of both functional representations through the use of polynomials in the Bernstein form.
The constraints
required to provide a valid extremal dependence are addressed in a straightforward manner, by placing a prior on the coefficients of the Bernstein polynomials which gives probability one to the set of valid functions. The prior is extended to the polynomial degree, making our approach fully nonparametric.
Although the analytical expression of the posterior is unknown, inference is possible via a trans-dimensional MCMC scheme.
We show the efficiency of the proposed methodology 
by means of a simulation study.   
The extremal behaviour of log-returns of daily
exchange rates between the Pound Sterling vs the U.S. Dollar and the Pound Sterling vs the Japanese Yen is analysed for illustrative purposes.

\end{abstract}



\received{\smonth{12} \syear{2015}}

\tableofcontents

\end{frontmatter}

\maketitle

\section{Introduction}\label{sec:intro}

The estimation of future extreme episodes of a real process, such as heavy-rainfall, heat-waves and simultaneous losses in the financial market, is of crucial importance for risk management. In most applications, an accurate assessment of such types of risks requires an appropriate modelling and inference of the dependence structure of multiple extreme values.

A simple definition of multiple extremes is obtained by applying the definition of block (or partial)-maximum \cite[Ch. 3]{coles01} to each of the variables considered.
Then, the probabilistic modelling concerns the joint distribution of the random vector of so-called component-wise (block) maxima, 
in short {\it sample maxima}, whose joint distribution is named a 
\emph{multivariate extreme value distribution} \cite[Ch. 6]{dehaan+f06}. 
Within this approach, parametric models for the dependence structure have been
widely discussed and applied in the literature 
(e.g. \citeNP{coles01}, \citeNP{boris+p2015}), but a major downside is that a model which may be useful for a specific application is often too restrictive for many others. 
As a consequence, more recently, much attention has been devoted to the study of nonparametric estimators or estimation methods for assessing the extremal dependence (see e.g. \citeNP[Ch. 7]{dehaan+f06}). 
Some examples focused on nonparametric estimators of the Pickands dependence function \cite{pickands81} are provided in \shortciteN{cap+f+g97}, 
\citeN{genest2009rank}, \shortciteN{bucher2011},
\shortciteN{berghaus2013} and \shortciteN{marcon+p+n+m15}, among others. 
Examples of Bayesian modelling of the extremal dependence are \citeN{boldi+d07}, \citeN{guillotte2008} and \citeN{sabourin2014} to cite a few.

In order to provide a comprehensive discussion of our approach, we restrict our attention to the bivariate case, that is to two-dimensional vectors of sample maxima.
Specifically, we describe how Bernstein polynomials (\citeNP{lorentz53}) can be used to model the extremal dependence within a Bayesian nonparametric framework. 
In recent years, Bernstein polynomials are attracting much attention
in Bayesian nonparametric statistics, in that they are useful as prior distributions of distribution functions (\citeNP{petrone99}), 
for density estimation (\citeNP{petrone1999bayesian}) and they have nice properties (\citeNP{ghosal2001}, \citeNP{petrone2002}).

Our present proposal has the following key features  that make it different from \shortciteN{marcon+p+n+m15}.
Firstly, the use of this particular polynomial expansion makes it possible to accommodate different representations of the dependence structure, such as the Pickands  dependence function and the so-called angular (or spectral) measure. This ensures that in each case, there is the fulfillment of some specific constraints which guarantee that a proper extreme value distribution is defined. Secondly, 
model fitting, inference and model assessment can be achieved via MCMC methods, preserving the relation between both extremal dependence forms.
Information about the polynomial degree is yielded from the data as part of the inferential procedure, and there is no need for a preliminary estimate as is often the case when regularization methods are applied, (e.g. \shortciteNP{fil+g+s08}, \shortciteNP{marcon+p+n+m15}). Additionally, there is no need to choose between representing the dependence by means of the angular measure or the Pickands dependence function.
Finally, the expression of approximate probabilities for simultaneous exceedances can be derived in closed-form and this implies that the predictive probability for such events is easy to calculate.

The paper is organised as follows. 
In Section~\ref{sec:extdep} we briefly describe some basic concepts regarding the extremal dependence structure.
In Section~\ref{sec:bayeisan} we propose a Bayesian nonparametric model for the extremal dependence along with an MCMC approach for posterior simulation. Section~\ref{sec:simu} illustrates the flexibility of the proposed approach by estimating the dependence structure of data simulated from some
popular parametric dependence models.
Section ~\ref{sec:app} provides a real data application, in which we analyse the exchange rates of the Pound Sterling against the U.S. Dollar and Japanese Yen, jointly, at extremal levels during the past few decades.

\section{Extremal Dependence}\label{sec:extdep}

In this section, we present some main ideas regarding multivariate extreme value theory, which we use for the development of the framework we propose. For more details see e.g. Chapters 4, 6 and 8 of \shortciteN{falk+h+r10}, \shortciteN{dehaan+f06} and \shortciteN{beirlant+g+s+t04}, respectively.

Assume that $\bZ=(Z_1,Z_2)$ is a bivariate random vector of sample maxima with an extreme value distribution $G$.
A distribution as such has the attractive feature of being {\it max-stable}, that is for all $n=1,2,\ldots$, there exist sequences 
of constants $a_n,c_n>0$ and $b_n,d_n\in\real$ such that $G^n(a_nz_1+b_n,c_nz_2+d_n)=G(z_1,z_2)$, for
all $z_1,z_2\in\real$.
Hereafter, we refer to $G$ as a bivariate max-stable distribution.
In particular, the margins of $G$,
denoted by $G_i(z)=\prob(Z_i\leq z)$, for all $z\in \real$ and $i=1,2$, 
are members of the Generalised Extreme Value (GEV) distribution \cite[Ch. 3]{coles01}, i.e.
\begin{equation}\label{eq:gev}
G_i(z_i;\mu_i,\sigma_i,\xi_i)=\exp
\left\{
-\left(
1+\xi_i\frac{z_i-\mu_i}{\sigma_i}
\right)_+^{-1/\xi_i}
\right\},
\end{equation}
where $z_i,\mu_i,\xi_i\in\real$, $\sigma_i>0$ for $i=1,2$ and $(x)_+=\max(0,x)$ and,
hence, are univariate max-stable distributions.
Taking the transformation, with the marginal parameters assumed to be known,
\begin{equation}\label{eq:gev_frechet}
Y_i=\left(1+\xi_i\frac{Z_i-\mu_i}{\sigma_i}
\right)_+^{-1/\xi_i},\quad i=1,2,
\end{equation}
then, the marginal distributions of $\bY=(Y_1,Y_2)$ are unit Fr\'{e}chet, i.e. $\prob(Y_i\leq y)=e^{-1/y}$, for all $y>0$ with $i=1,2$,
and the bivariate max-stable distribution takes the form
\begin{equation}\label{eq:bivgev}
G_0(y_1,y_2)=\exp\{-L(1/y_1,1/y_2)\}, \quad y_1,y_2>0, 
\end{equation}
where $L : [0, \infty)^2 \to [0, \infty)$, named the stable-tail dependence function \cite[pp. 221--226]{dehaan+f06} is given by
\begin{equation}\label{eq:stabletail}
L(x_1,x_2)=2\int_\spa \max\{x_1 \, w, x_2 \, (1-w)\} H(\diff w),\quad x_1,x_2\geq 0.
\end{equation}
$\spa=[0,1]$ denotes the one-dimensional simplex and $H$, named the angular (or spectral) measure, is the distribution function of
a probability measure supported on
$\spa$ 
and satisfying the following condition,
\begin{enumerate}
	\item[(C1)] The center of the mass of $H$ must be at $1/2$, that is,
\begin{equation*}
\int_{\spa} w\,H(\diff w)=\int_{\spa} (1-w)\,H(\diff w)=1/2.
\end{equation*}
\end{enumerate}
We stress that marginal parameters can always be estimated separately using some
standard methods (e.g. \citeNP[Ch. 3]{dehaan+f06}, \citeNP[Ch. 3,9]{coles01}) and 
hence be used to achieve the representation \eqref{eq:bivgev}.

More precisely, for any max-stable distribution $G_0$ there exists
a finite measure, $H^*$ on $\spa$,
satisfying the mean conditions $\int_{\spa} w\,H^{*}(\diff w)=\int_{\spa} (1-w)\,H^{*}(\diff w)=1$, which implies 
$H^{*}(\spa)=2$, 
such that $G_0$ can be represented by the general form \eqref{eq:bivgev}, where
the angular measure is given by the normalization $H:=H^{*}/H^{*}(\spa)$. We will use $H$ to denote both the probability measure and its distribution function, since the difference can be derived from the context.
Conversely, any probability measure with distribution function $H$ satisfying (C1) generates a valid bivariate max-stable distribution
\cite[Ch. 6]{dehaan+f06}.
As usual practice, for simplicity we focus on a subset of all the possible angular measures (\shortciteANP{beirlant+g+s+t04},  \citeyearNP{beirlant+g+s+t04}, Ch.\ 8).
\begin{assum}\label{ass:measure_cond}
Let $(\{0\},\spaint,\{1\})$ be a partition of $\spa$, where $\spaint=(0,1)$. 
Consider angular measures of the form $H([a,b])=p_0\delta_0([a,b])+ \Delta((a,b]) + p_1\delta_1([a,b])$,
for any $a,b\in\spa$ with $a\leq b$ and $p_0, p_1\in[0,1/2]$. Specifically, $\delta_x(A)$ is the Dirac measure
for any $x\in\real$ and a measurable set $A\subset\real$, $\Delta((a,b])=\Hint(b)-\Hint(a)$ is the Lebesgue-Stieltjes measure,
where  $\Hint(w)=\int_0^w h(t) \, \diff t$ and $h(t)\geq 0$ is a Lebesgue integrable function such
$\int_0^1 h(w) \, \diff w=1-p_0-p_1$.
\end{assum}

The role of the angular measure can be explained by means of its geometric interpretation. The more the dependence between variables increases (the more likely it is that they are similar in value), the more the mass of $H$ tends to accumulate at the center of the simplex, i.e. $1/2$ by condition (C1). Conversely, the more the mass of $H$ moves to the vertices of the simplex, the more the variables become independent.
The distribution function of the angular measure is
\begin{equation}\label{eq:angular_dist}
H([0,w])=p_0+\Hint(w)+p_1\indic_{[0,w]}(1),\quad w\in\spa
\end{equation}
where $\indic_A(x)$ is the indicator function of the set $A$.
This means that $H$ has atoms on the vertices $\{0\}$ and $\{1\}$, denoted by $p_0=H(\{0\})$ and 
$p_1=H(\{1\})=H([0,1])-H([0,1))$ respectively, 
and it is absolutely continuous on $\spaint$.
Notice that, by the mean constraint (C1), the following two identities must be satisfied
\begin{equation}\label{eq:mass_vertices}
p_1=1/2-\int_0^1 w \, h(w) \, \diff w,\quad p_0=1/2-\int_0^1(1-w) \, h(w) \, \diff w.
\end{equation}
We stress that although \eqref{eq:angular_dist} excludes atoms in $\spaint$, it is already rich enough to describe the dependence of many practical applications. 
In the following sections, we will denote by $\spH$ the space of angular distributions defined in this way, so that each $H\in\spH$ is defined by a valid triplet $(p_0,p_1,\Hint)$. 

The properties of the stable-tail dependence function 
are: 
a) it is homogeneous of order $1$, that is $L(vx_1,vx_2)=vL(x_1,x_2)$
for all $v,x_1,x_2>0$; b) $L(x,0)=L(0,x)=x$ for all $x>0$; c) it is continuous and convex, i.e. 
$L(v(x_1,x_2) + (1-v)(x'_1,x'_2))\leq vL(x_1,x_2)+(1-v)L(x'_1,x'_2)$ for all $x_1,x_2,x'_1,x'_2\geq 0$ and $v\in\spa$;
d) $\max(x_1,x_2) \leq L(x_1,x_2) \leq x_1+x_2$ for all $x_1,x_2\geq0$.  
The lower and upper bounds of the last
condition represent the cases of complete dependence and independence, respectively. 
By the homogeneity of $L$ we have that, for all $x_1,x_2\geq 0$,
\begin{equation}\label{eq:pickands}
L(x_1,x_2)=(x_1+x_2)A(t),\quad A(t)=2\int_\spa \max\{t\, (1-w),(1-t)\,w\}H(\diff w),
\end{equation}
where $t=x_2/(x_1+x_2)\in\spa$. The function $A$ is called the Pickands dependence function and, by the properties of $L$, it satisfies the following conditions:
\begin{enumerate}
\item[(C2)] $A(t)$ is convex, i.e., $A(at+(1-a)t')\leq aA(t)+(1-a)A(t')$,
for $a, t,t'\in\spa$;
\item[(C3)]  $A(t)$ has lower and upper bounds
$$
1/2\leq \max\left(t,1-t \right) \leq A(t) \leq 1;\quad t\in\spa.
$$
\end{enumerate}
In condition (C3),  the lower and upper bounds represent  the cases of complete dependence and independence, respectively. In other words, any Pickands dependence function belongs to the class $\spA$ of functions $A: \spa \rightarrow [1/2,1]$ satisfying the above conditions (\shortciteANP{falk+h+r10},  \citeyearNP{falk+h+r10}, Ch.\ 4). Conversely, if a function $A\in\spA$ has second derivatives on $\spaint$, then a valid angular measure $H$ exists, such that
\begin{equation}\label{eq:relation_pick_ang}
A(t)=1+2\int_0^tH(w) \, \diff w-t\quad t\in\spa
\end{equation}
and therefore $A'(t)=-1+2H([0,t])$, where $A'$ is seen as the right-hand derivative and $A''(t)=2h(t)$, for $t\in\spaint$ 
(\shortciteANP{beirlant+g+s+t04},  \citeyearNP{beirlant+g+s+t04}, Ch.\ 8). 
From the above relation follows that the atoms on the vertices of the simplex can be expressed by the Pickands dependence function
as $p_0=\{1+A'(0)\}/2$ and $p_1=\{1-A'(1)\}/2$, where $A'(1)=\sup_{t\in[0,1)}A'(t)$.

The angular distribution is also used to define another important tail dependence function, $R$, given by 
\begin{equation*}\label{eq:taildep}
R(x_1,x_2)=2\int_\spa \min\{x_1\, w,x_2 \, (1-w)\} H(\diff w),\quad x_1,x_2\geq 0,
\end{equation*}
or equivalently, by $R(x_1,x_2)=x_1+x_2-L(x_1,x_2)$.
This function can be used to approximate the
probability of simultaneous exceedances, i.e.
\begin{equation}\label{eq:simu_excess}
\prob(Y_1>y_1, Y_2>y_2)\approx R(1/y_1,1/y_2),
\end{equation}
for high enough thresholds $y_1,y_2>0$ (e.g., \citeNP{boris+p2015}), as well as to compute the coefficient of upper tail
dependence (e.g., \citeNP[p.163]{coles01}), i.e.
\begin{equation}\label{eq:coef_up_tail} 
\chi = \lim_{y\rightarrow + \infty} \prob(Y_1>y | Y_2>y)
= \lim_{y\rightarrow + \infty} \prob(Y_2>y | Y_1>y) \equiv R(1,1)\in[0,1].
\end{equation}
This is an important summary measure of the extremal dependence between two random variables. $Y_1$ and $Y_2$ are independent in the upper tail when $\chi=0$, whereas they are completely dependent when
$\chi=1$.
%

\section{Bayesian nonparametric modeling of $H$ and $A$}\label{sec:bayeisan}

%
\subsection{Bernstein Polynomial Representation}\label{subsec:bern}

The basic idea behind our proposal is to define both the distribution function of the angular measure and the Pickands dependence function as polynomials, restricted to $\spa$, of the form $\sum_{j=0}^k\,a_j b_j(x)$, where each $a_j$ is a real-valued coefficient and the $b_j(\cdot)$, $j=1,2,\ldots$ form an adequate polynomial basis. 
Denote by $\spP_k$ the space of polynomials of degree $k$, and let $\spH$ and $\spA$ be the sets of angular distributions and Pickands 
dependence functions, respectively, as in the previous section.
Since $\bigcup_{k=0}^\infty\spP_k$ is dense in the spaces $\spH$ and $\spA$, we know that any angular distribution function in $\spH$ as well as any Pickands dependence function in $\spA$, can be arbitrarily well approximated
by a polynomial in $\spP_k$ for some $k$. 
Due to their shape preserving properties, it is convenient to use a Bernstein polynomial basis (\citeNP{lorentz53}) that, when restricted to $\spa$, will allow us to construct proper functions on $\spH$ and $\spA$ by identifying valid sets of coefficients.

For each $k=1,2,\ldots$, the Bernstein basis polynomials of degree $k$ are defined as
\begin{eqnarray*}\label{eq:bp}
b_j( x; k) =\frac{k!}{j! (k-j)!}\,x^{j}(1-x)^{k-j}, \quad j=0,\ldots,k.
\end{eqnarray*}
Throughout the article, use will be made of the simple identities,
\begin{equation}\label{eq:basis_beta}
(k+1) \, b_j( x; k)=\betaf(x|j+1,k-j+1),\quad x\in\spa,
\end{equation}
where $\betaf(\cdot|a,b)$ denotes the beta density function with shape parameters $a,b>0$, and for $k\geq1$
\begin{equation}\label{eq:basis_conrners}
b_j(0;k)=\delta_{j,0},\quad b_j(1;k)=\delta_{j,k},
\end{equation}
where $\delta_{j,r}$ is Kronecker delta function (e.g., \citeNP{petrone99}).

We start modeling the extremal dependence by representing the distribution function \eqref{eq:angular_dist} through
a polynomial of degree $k-1$ in Bernstein form, for some $k=1,2,\ldots$. Specifically, we define
\begin{equation}\label{eq:bpoly_angdist}
 H_{k-1}([0,w]) := 
 \left\{
\begin{tabular}{lcl}
$\sum_{j\leq k-1} \eta_j \;  b_j(w; k-1)$ &  if  & $w\in[0,1)$\\
1 & if & $w=1$\\ 
 \end{tabular}
 \right.
\end{equation}
and taking the first derivative of $ H_{k-1}$ with respect to $w$, we have that the density in the interior of $\spa$ is equal to
\begin{equation}\label{eq:bpoly_den}
H'_{k-1}(w)=\sum_{j=0}^{k-2} (\eta_{j+1}-\eta_j) \, \betaf(w|j+1,k-j-1)=:h_{k-1}(w),\quad w\in\spaint
\end{equation}
\begin{prop}\label{prop:angular_dist}
By forcing the coefficients $\eta_0,\ldots,\eta_{k-1}$ in \eqref{eq:bpoly_angdist}, for fixed polynomial degree $k$, to meet the restrictions:
\begin{enumerate}
\item[(R1)] $0\leq p_0=\eta_0\leq \eta_1\leq \ldots \leq \eta_{k-1}=1-p_1\leq 1$; 
\item[(R2)] $\eta_0+\cdots+\eta_{k-1}=k/2$;
\end{enumerate}
it is ensured that $H_{k-1}$ is the distribution function of a valid angular measure satisfying Assumption \ref{ass:measure_cond}.
\end{prop}

Alternatively, we can also model the extremal dependence by representing the Pickands dependence function in \eqref{eq:pickands}
with a polynomial of degree $k=0,1,\ldots$ in the Bernstein form. Specifically, let
\begin{equation}\label{eq:bpoly_picka}
 A_{k}(t) := \sum_{j=0}^k \beta_j b_j(t; k),
 \qquad t \in\spa,
\end{equation}
then by forcing the coefficients $\beta_0,\ldots,\beta_k$ in \eqref{eq:bpoly_picka} to meet the restrictions:
\begin{enumerate}
\item[(R3)] $\beta_0=\beta_k=1\ge \beta_j,$ for all $j=1,\ldots,k-1$;
\item[(R4)] $\beta_1=\frac{k-1 + 2 p_0}{k}$ and $\beta_{k-1} = \frac{k-1 + 2p_1}{k}$;
\item[(R5)] $\beta_{j+2}-2\beta_{j+1}+\beta_j\geq 0$, $j=0,\ldots,k-2$;
\end{enumerate}
it is ensured that $A_{k}$ satisfies conditions (C2)-(C3) and hence it is a proper Pickands dependence function (\shortciteANP{marcon+p+n+m15},  \citeyearNP{marcon+p+n+m15}). This is easily explained by the following. First, by \eqref{eq:basis_conrners} we have that $A_k(0)=A_k(1)=1$
if $\beta_0=\beta_k=1$ and it is immediate to check that $A_k(t)=1$ for all $t\in\spa$ when $\beta_0=\cdots=\beta_k=1$. Because $b_j(t; k)\leq 1$
for all $t\in\spa$ then $A_k(t)\leq 1$ by (R3).  Second, $A_k(t)\geq\max(t,1-t)$ for all $t\in\spa$ if $A_k'(0)\geq -1$ and $A_k'(1)\leq 1$, where
\begin{equation}\label{eq:bpoly_picka_fd}
 A'_{k}(t) = \sum_{j=0}^{k-1} (\beta_{j+1}-\beta_j) \, \betaf(t|j+1,k-j), \quad t \in\spa.
\end{equation}
Since $A_k'(0)=k(\beta_1-\beta_0)$, $A_k'(1)=k(\beta_k-\beta_{k-1})$ and, on the other hand, knowing also from \eqref{eq:relation_pick_ang} that
$A_k'(0)=2p_0-1$ and $A_k'(1)=1-2p_1$, we obtain the conditions in (R4), which imply that
$\beta_1\geq 1-1/k$ and $\beta_{k-1}\geq 1-1/k$. Finally, $A_k(t)$ is convex if $A''_k(t)\geq 0$ for all $t\in\spa$, where
\begin{equation}\label{eq:bpoly_picka_sd}
 A_{k}^{''}(t) = k \,  \sum_{j=0}^{k-2} (\beta_{j+2}-2\beta_{j+1}+\beta_j) \, \betaf(t|j+1,k-j-1), \quad t \in\spa.
\end{equation}
Clearly the positivity of \eqref{eq:bpoly_picka_sd} is guaranteed by the conditions in (R5).

Under Assumption \ref{ass:measure_cond}, the distribution function \eqref{eq:bpoly_angdist} and the Pickands dependence function
\eqref{eq:bpoly_picka} are linked, as described by the next result.
\begin{prop}\label{prop:equivalence_AH}
Let $H_{k-1}$ be the distribution function of an angular measure with expression \eqref{eq:bpoly_angdist},  and
$A_k $ be the Pickands dependence function given in \eqref{eq:bpoly_picka}. 
Then, the following are equivalent:
\begin{itemize}
\item[i)]  Given $A_{k}$ one may recover $H_{k-1}$ by means of their coefficients' relationship:
\begin{equation}\label{eq:etas}
\eta_j=\frac{k}{2}\left(\beta_{j+1}-\beta_j+\frac{1}{k}\right),\quad j=0,\ldots,k-1.
\end{equation}
Conversely, given $H_{k-1}$, one may recover $A_{k}$ by means of their coefficients' relationship:
\begin{equation}\label{eq:betas}
\beta_{j+1}=\frac{1}{k}\left(2\sum_{i=0}^{j}\eta_i + k-j-1\right),\quad j=0,\ldots,k-1,
\end{equation}
with $\beta_0=1$.
\item[ii)] Restrictions (R1) and (R2) are satisfied and $H_{k-1}$ meets condition (C1), if and only if 
restrictions (R3)-(R5) are verified and $A_k$ meets conditions (C2) and (C3). 
\end{itemize}
\end{prop}
This result tells us that the one-to-one relationship between the angular measure and the Pickands dependence function, when these are represented with Bernstein polynomials, is simply expressed through a one-to-one relationship between their corresponding coefficients.
Its implications are as follows. One can estimate the coefficients in \eqref{eq:bpoly_angdist} so that they meet the conditions (R1) and (R2) and then compute the coefficients in \eqref{eq:bpoly_picka} by equation \eqref{eq:betas}, which will automatically meet the conditions (R3)-(R5). 
Or vice versa, estimate the coefficients in \eqref{eq:bpoly_picka} satisfying conditions (R3)-(R5) and derive the coefficients in \eqref{eq:bpoly_angdist} 
by \eqref{eq:etas}, which will satisfy (R1) and (R2).
As a consequence, for inference, there is no need to choose between one or the other way of representing the dependence structure.

Finally, $H_{k-1}$ and $A_k$ can provide accurate approximations of the true functions $H$ and $A$.
\begin{prop}\label{prop:approximation_H_A}
Let 
\begin{align*}
\spH_{k-1}=\{ & w\mapsto H_{k-1}(w)= \sum_{j\leq k-1}\eta_j \, b_j(w;k-1):  \\ 
& \eta_0,\dots,\eta_{k-1}\in[0,1]  \text{ and (R1)-(R2) are satisfied}\}
\end{align*}
and
$$
\spA_k=\{t\mapsto A_k(t)=\sum_{j\leq k}\beta_jb_j(t;k): \beta_0,\dots,\beta_{k}\in[0,1] \text{and (R3)-(R5) are satisfied}\}.
$$
Then, $\spA_k$ and $\spH_{k-1}$, $k=1,2,\ldots$ are nested sequences in $\spA$ and $\spH$, respectively. Additionally, there are polynomials
$A_k$ and $H_{k-1}$ such that 
\begin{equation}\label{eq:convergence_A}
\lim_{k\rightarrow\infty}\sup_{t\in\spa} |A_k(t)-A(t)| =0
\end{equation}
and
\begin{equation}\label{eq:convergence_H}
\lim_{k\rightarrow\infty}\sup_{w\in\spa} |H_{k-1}(w)-H(w)| =0.
\end{equation}
\end{prop}
\subsection{Bayesian Inference}\label{subsec:prior}

We provide details of the key ingredients of a Bayesian nonparametric model for the extremal dependence. This can be formulated through \eqref{eq:bpoly_angdist} or \eqref{eq:bpoly_picka}, indifferently since, as seen in Section \ref{subsec:bern}, one expression can always be recovered from the other. We show the explicit forms in which the prior distribution and the likelihood function for one approach are linked to those of the other.

We start by constructing a prior probability on the space $\spH$ of valid angular measures using the Bernstein polynomial representation \eqref{eq:bpoly_angdist}, for some polynomial order $k$.
Then, the prior on $\spH$ is induced by a joint prior distribution on $(k,\bdeta_k)$, where $\bdeta_k=(\eta_0,\ldots,\eta_{k-1})$. Because $\eta_0=p_0$ and $\eta_{k-1}=1-p_1$
by (R2), we conveniently express the prior distribution as
\begin{equation}\label{eq:joint_prior}
\Pi(k, \bdeta_k)=\Pi(\bdeta_k|k)\,\Pi(k).
\end{equation}
Note that for $k<3$, the resulting dependence structure is trivial, so we will only consider the case when $k\geq 3$.
Some convenient choices for the prior distribution of the polynomial order are   $\Pi(k)=\text{Pois}(k-3|\kappa_P)$ or $\Pi(k)=\text{nbin}(k-3|\kappa_{NB},\sigma^2)$, where $\kappa_P,\kappa_{NB}>0$ are the means of Poisson and negative binomial distributions, respectively. The latter, however, is more flexible through its variance $\sigma^2$. Specifically, the probability mass function for the negative binomial distribution is $\Gamma(x+s)/(\Gamma(s) x!) \; p^{s} \; (1-p)^x$, for $ x = 0, 1, 2, \ldots$, with target for number of successful trials $s > 0$ and probability of success in each trial $0 < p \leq 1$. With this parametrization, the mean corresponds to $\kappa_{NB} = s(1-p)/p$ and variance $\sigma^2 = s(1-p)/p^2$.
In order to define a valid prior on $\spH$, $\Pi(\bdeta_k|k)$ must assign, for each $k\in\nat$, probability one to the set $\spE=\spE(k)\subset \spa^{k}$ of  $k$-dimensional vectors satisfying (R1) and (R2). 
By (R1) we have that the atoms on the edges are represented by parameters $\eta_0=p_0$ and $\eta_{k-1}=1-p_1$. Given the particular role that these quantities play in the model, and the relevance of their interpretation, we have decided to treat them separately when defining the prior. Furthermore, this choice seems empirically justified by the results obtained through simulation studies.
Therefore, we define the conditional prior for the polynomial coefficients given the degree $k$ in the following manner, 
$$
\Pi(\bdeta_k|k)=\Pi(\eta_{1},\ldots,\eta_{k-2}|p_1,p_0,k)\, \Pi(p_1|k,p_0)\,\Pi(p_0).
$$ 
Specifically, we let
$\Pi(p_0)=\text{Unif}(0,1/2).$ Then,
\begin{equation}\label{eq:cond_p0}
	(k-1) \, p_0+(1-p_1)\le\sum_{j=0}^{k-1}\eta_j= k/2\le p_0 + (k-1)(1-p_1),
\end{equation} 
where the identity follows from condition (R2), while the two inequalities stem from (R1). After simple manipulations, it follows that, in order for (R1) and (R2) to hold, a necessary condition is $(k-1) \, p_0-k/2+1\le p_1\le (p_0+k/2-1)/(k-1)$, so we set $\Pi(p_1|k, p_0)=\text{Unif}(a,b)$, with interval limits given by $a=a(k,p_0)=\max\{0, (k-1)p_0-k/2+1\}$ and $b=b(k,p_0)=(p_0+k/2-1)/(k-1)$. 

Now, conditional on $k$, $\eta_0$ and $\eta_{k-1}$, we set $X_0=\eta_0$, $X_{k-1}=\eta_{k-1}$ and we extend the prior distribution to the remaining parameters $\eta_1,\ldots,\eta_{k-2}$ by focusing on  the differences $X_j=\eta_{j}-\eta_{j-1}$, $j=1,\ldots,k-1$, in order to guarantee that condition (R1) is satisfied. For simplicity, analogous to what we did with $p_1$, we make such differences conditionally uniformly distributed on appropriate intervals, specified below, in order to satisfy also condition (R2), that is
%
%
\begin{equation}\label{eq:mconst_diff}
\sum_{j=0}^{k-1}\eta_j=\sum_{j=0}^{k-1}(k-j)X_j = k/2.
\end{equation}
Notice that we can rewrite \eqref{eq:mconst_diff} as
$$
(k-j)X_j + \sum_{l=j+1}^{k-1}(k-l)X_l = k/2 -\sum_{l=0}^{j-1}(k-l)X_l,
$$
for $j=1,\ldots,k-2$. Thus, if we assume that $X_l=0$ for $l=j+1,\ldots,k-2$, so that $\eta_l=\eta_j$ for
$l=j+1,\ldots,k-2$, we attain the upper bound,
$$
X_j \leq \frac{1}{k-j-1}\left(k/2 +p_1 -1 -\sum_{l=0}^{j-1}(k-l-1)X_l\right),
$$
for $j=1,\ldots,k-2$. On the other hand, if we assume that 
$$
\sum_{l=j+1}^{k-2}(k-l-1)X_l=(k-j-2)\bigg(1-p_1-\sum_{l=0}^j X_l\bigg),
$$
corresponding to $\eta_l=1-p_1$ for $l=j+1,\ldots,k-2$, then we attain, through few algebraic manipulations, the lower bound
$$
X_j \geq \max\left\{0, k/2 +(j-k+1)(1-p_1) - \sum_{l=0}^{j-1} (j-l+1)X_l\right\},
$$
for $j=1,\ldots,k-2$.
%
Rewriting these inequalities in terms of $\eta_j$, we find that the widest valid range for the coefficients can be expressed in terms of intervals $\spE_j=\spE_j(k,\eta_0,\ldots,\eta_{j-1},\eta_{k-1})$, given by
\begin{align*}
 \spE_j=\Bigg[
 \max \Bigg\{ \eta_{j-1}, \; & \frac{k}{2} + (k-j-1)(p_1-1) - \sum_{l=0}^{j-1}\eta_l \Bigg\}; \; \\
 & \min \Bigg\{ 1-p_1; \frac{1}{k-j-1}\Big(\frac{k}{2} +p_1 -1 -\sum_{l=0}^{j-1}\eta_l\Big)\Bigg\}
\Bigg],
\end{align*}
for $j=1,\ldots,k-2$.
Finally, we let $\eta_j|(k,\eta_0,\ldots,\eta_{j-1},\eta_{k-1})$ for $j=1,\ldots,k-2$ be
conditionally independent and uniformly distributed on such intervals, therefore arriving at the following conditional prior distribution
\begin{align}\label{eq:prior_eta}
\Pi(\eta_{1},\ldots,\eta_{k-2}|k,p_1,p_0)=\prod_{j=1}^{k-2}\Pi(\eta_j |k, \eta_0,\ldots,\eta_{j-1},\eta_{k-1})=\prod_{j=1}^{k-2}\text{Unif}(\spE_j).
\end{align}
A direct consequence of Proposition \ref{prop:equivalence_AH} is that a valid prior distribution is induced also
on the space $\spA$ of valid Pickands dependence functions, as expressed by the following result. 
\begin{cor}\label{cor:prior_on_beta}
Let $\setB=\setB(k)\subset \spa^{k+1}$ be the space of $(k+1)$-dimensional vectors 
satisfying restrictions (R3)-(R5). Then, for any fixed $k\ge 3$ the prior distribution 
\eqref{eq:prior_eta} induces a prior distribution on the coefficients of $A_k$ in \eqref{eq:bpoly_picka}. 
Precisely $\beta_j|(\beta_0,\ldots,\beta_{j-1}, \beta_{k-1})$, for $j=2,\ldots,k-2$, 
turns out to be conditionally independent and uniformly distributed
on the intervals
\begin{align*}
\setB_j=\Bigg[
\max \Bigg\{ 2\,\beta_{j-1}-\beta_{j-2}, \; & (k-j) \beta_{k-1} - (k-j-1)\Bigg\}; \; \\
& \frac{1}{k-j} \bigg( \beta_{k-1} + (k-j-1) \beta_{j-1} \bigg) \Bigg].
\end{align*}
The prior distribution on $\bbeta_k = (\beta_0, \ldots, \beta_k)$ is then given by 
\begin{eqnarray*}
\Pi(\bbeta_k|p_1, p_0, k)&=&\indic_{\{1\}}(\beta_0) \; \indic_{\{(k-1+2 \, p_0)/k\}}(\beta_1) \prod_{j=2}^{k-2}\Pi(\beta_j|\beta_0,\ldots,\beta_{j-1},\beta_{k-1}) \\
&&\\
&\times& \indic_{\{(k-1+2 \, p_1)/k\}}(\beta_{k-1}) \; \indic_{\{1\}}(\beta_k)\\
&& \\
&=&\indic_{\{1\}}(\beta_0) \; \indic_{\{(k-1+2 \, p_0)/k\}}(\beta_1) \; \indic_{\{(k-1+2 \, p_1)/k\}}(\beta_{k-1}) \; \indic_{\{1\}}(\beta_k) \; \\
&&\\
&\times& \prod_{j=2}^{k-2} \text{Unif} (\setB_j) \; \left(\frac{k}{2}\right)^{k-3}.
\end{eqnarray*}
\end{cor}
This result follows directly from the change of variable formula. In fact, letting $\beta(\eta ; \, k)$ given by expression \eqref{eq:betas} denote the inverse transformation of $\eta(\beta; \, k)$, given by expression \eqref{eq:etas}, the corresponding Jacobian is $(k/2)^{k-3}$. Notice that, in this representation, the point masses of $H$ are given by $p_0=1/2-k(1-\beta_1)/2$ and $p_1=1/2-k(1-\beta_{k-1})/2$.

The prior thus constructed assigns positive probability to any subset of $(\{k\}\times \spa^{k})$, $k>1$ which is valid, in the sense of satisfying conditions (R1)-(R2) or, equivalently, to every subset of $(\{k\}\times \spa^{k+1})$, $k>1$ which is valid, in the sense of satisfying conditions (R3)-(R5). It therefore follows from proposition \ref{prop:approximation_H_A} that the prior has a full support, in terms of the $L_\infty$ norm, on the spaces $\spH$ and $\spA$.

We now derive the analytical expression of the likelihood function. To do so, we consider for simplicity the distribution \eqref{eq:bivgev} with stable tail 
dependence function represented by \eqref{eq:pickands}.
Then, the joint probability density function (p.d.f.) is given by
$$
g(y_1,y_2)=|J(y_1,y_2)| \; \frac{\partial^2}{\partial x_1 \partial x_2} G(1/x_1,1/x_2)\Big\vert_{x_1=1/y_1, x_2=1/y_2},
$$
for all $y_1,y_2>0$, where $J(y_1,y_2)=(y_1y_2)^{-2}$. This 
is equal to
\begin{equation*}\label{eq:joint_den}
g(y_1,y_2)=G(y_{1}, y_{2})\left[ \frac{\left\{ A(t) - t\, A'(t) \right\} \left\{ A(t) + (1-t) \, A'(t) \right\}}
{(y_1y_2)^2} + \frac{A''(t)}{(y_1+y_2)^3} \right].
\end{equation*}
Let $\bY_{1:n}=(\bY_1,\ldots,\bY_n)$ be i.i.d. copies of a bivariate max-stable random vector with p.d.f. $g(y_1,y_2)$. 
Assume that the Pickands dependence function is represented by \eqref{eq:bpoly_picka}, for some fixed $k$. 
Then, the log-likelihood function is equal to
\begin{align}\label{eq:llik_beta} \nonumber
\ell(\by_{1:n}|\btheta) 
&= - \sum_{i=1}^{n} \left( \frac{1}{y_{1,i}} + \frac{1}{y_{2,i}} \right) \sum_{j=0}^{k} \beta_j \; b_j(t_i;k)\\ \nonumber
& \\ \nonumber
&+ \sum_{i=1}^{n} \log\Bigg\{ \, \Big( \sum_{j=0}^{k} \beta_j \; b_j(t_i;k) - t_i \, k \, \sum_{j=0}^{k-1} (\beta_{j+1} - \beta_j) \; b_j(t_i;k-1) \Big)\\ \nonumber
& \\ \nonumber
&\times \frac{\sum_{j=0}^{k} \beta_j \; b_j(t_i;k) + (1 - t_i) k \, \sum_{j=0}^{k-1} (\beta_{j+1} - \beta_j) \; b_j(t_i;k-1)} 
{(y_{1,i}y_{2,i})^2}\\ \nonumber
& \\ 
&+ \frac{k\, (k-1) \, \sum_{j=0}^{k-2} (\beta_{j+2} - 2 \, \beta_{j+1} + \beta_j) \; b_j(t_i;k-2)}{(y_{1,i}+y_{2,i})^3} \, \Bigg\},
\end{align}
where $\btheta=(k,\beta_0,\ldots,\beta_{k})\in\bTheta\subseteq (\nat\times \spa^{k+1})$. 
We denote by $\lik(\by_{1:n}|\btheta)$ the associated likelihood function.
We may once again apply Proposition \ref{prop:equivalence_AH},  to obtain the log-likelihood function in terms of 
$\btheta=(k,\eta_0,\ldots,\eta_{k-1})\in\bTheta\subseteq (\nat\times \spa^{k})$ which, abusing terminology, can be seen as a reparametrization. More formally, this corresponds to the representation of the distribution \eqref{eq:angular_dist}, in the stable-tail dependence function \eqref{eq:stabletail}, by means of a polynomial angular distribution given by the expression \eqref{eq:bpoly_angdist}.

There is no closed form for the posterior distribution $\Pi^n(\btheta|\by_{1:n})$ which is proportional to $\Pi(\btheta) \lik(\by_{1:n}| \btheta),$ regardless of the representation considered.
For this reason, we base the model inference on a complex MCMC 
posterior simulation scheme and, to be concise, we only describe the estimation procedure of the polynomial angular distribution, since it has been established that the Pickands dependence function can be obtained through a transformation.
The main difficulty stems from the fact that, at each MCMC iteration, the dimension of the vector of coefficients $\bdeta_k$ changes with $k$. 
We therefore resort to a trans-dimensional MCMC scheme proposed by \citeN{godsill2001} and, in the infinite-dimensional case, applied by \citeN{antoniano2013}. Thus, we extend $\Pi(k, \bdeta_k)$ to 
$$
\Pi(k, \bdeta_\infty) = \Pi ( \bdeta_k | k ) \; \Pi (k) \; \prod_{j>k} \Pi(\eta_j), 
$$
where $\bdeta_\infty=(\eta_0,\eta_1,\ldots)$ denotes an infinite sequence of which, given $k$ only the first $k$ elements are relevant, and $\Pi(\eta_j)$ is any fully known distribution.
In order to update the pair $(k^{(s)},\bdeta^{(s)}_\infty)$ at the current state $s$ of the Markov chain, we propose a Metropolis-Hastings step with the following proposal distribution,
$$
q(k,\bdeta_\infty | k^{(s)},\bdeta^{(s)}_\infty)=
q_k( k | k^{(s)})
\cdot q_\eta(\bdeta_k | k)
\cdot \prod_{ j>k} \Pi(\eta_j)
$$
where $q_\eta(\bdeta_k| k)$ coincides with the conditional prior $\Pi(\bdeta_k|k)$, 
$\Pi(\eta_j)$ is a fully specified density on $\spa$ and
$$
q_k\left(k=k^{(s)}+1|k^{(s)}\right) = 
\begin{cases}
1& \quad \text{if } k^{(s)}=3\\
1/2 & \quad \text{if } k^{(s)}>3\\
\end{cases}
$$
and 
$$
q_k\left(k=k^{(s)}-1|k^{(s)}\right) = 
\begin{cases}
0& \quad \text{if } k^{(s)}=3\\
1/2 & \quad \text{if } k^{(s)}>3.\\
\end{cases}
$$
Thus, given the current state $s$ of the Markov chain and the proposal indexed by $s+1$, the acceptance probability depends on the ratio
\begin{equation*}\label{eq:acceptanceprob}
p\left(k^{(s+1)},\bdeta_{k^{(s+1)}}^{(s+1)},k^{(s)},\bdeta_{k^{(s)}}^{(s)}\right)=
\frac{ \Pi^{n}(k^{(s+1)}, \bdeta^{(s+1)}_\infty| \by_{1:n}) \, q(k^{(s)},\bdeta_{\infty}^{(s)} | k^{(s+1)},\bdeta^{(s+1)}_{\infty})}
{\Pi^{n}(k^{(s)}, \bdeta_\infty^{(s)}| \by_{1:n}) \, q(k^{(s+1)},\bdeta^{(s+1)}_{\infty} | k^{(s)},\bdeta_{\infty}^{(s)})}
\end{equation*}
which, for any $k^{(s)}>3$, simplifies to 
$$
p\left(k^{(s+1)},\bdeta_{k^{(s+1)}}^{(s+1)},k^{(s)},\bdeta_{k^{(s)}}^{(s)}\right)=
%
\frac{\Pi(k^{(s+1)})}{\Pi(k^{(s)})}\;
\frac{\lik(\by_{1:n};k^{(s+1)},\bdeta_{k^{(s+1)}}^{(s+1)})}{\lik(\by_{1:n};k^{(s)},\bdeta^{(s)}_{k^{(s)}})},
$$
For $k^{(s)}=3$, we have $k^{(s+1)}=k^{(s)}+1$ with probability one, so there is a $1/2$ factor multiplying the ratio. 

This leads to the following algorithm.
\begin{algo}\label{algo:metro_etas}
MCMC scheme to draw samples from the posterior distribution $\Pi^{n}(k, \bdeta_{k}| \by_{1:n})$ of the polynomial order and coefficients.
\begin{enumerate}
\item Set $s=0$ and some starting values for the parameters
$\left( k^{(s)}, \bdeta_{k^{(s)}}^{(s)} \in \spE_{k^{(s)}}\right)$;
%
%
\item Repeat $M$ times the update of the parameters according to:
\begin{enumerate}
\item Draw the proposals: 
$$
k^{(s+1)}\sim q_k(k|k^{(s)}) \text{  and  } \bdeta_{k^{(s+1)}}^{(s+1)}\sim q_{\bdeta}(\bdeta_k| k^{(s+1)}, k^{(s)}, \bdeta_{k^{(s)}});
$$
\item Compute the acceptance probability:
$$
p = \min \Bigg(p\left(k^{(s+1)},\bdeta_{k^{(s+1)}}^{(s+1)},k^{(s)},\bdeta_{k^{(s)}}^{(s)}\right),\;1\Bigg);
$$
\item Draw  $U\sim\mbox{unif}\,(0,1)$ and if $U>p$ then set:
$$
\Big(k^{(s+1)},\,\bdeta_{k^{(s+1)}}^{(s+1)}\Big)=\Big(k^{(s)},\,\bdeta_{k^{(s)}}^{(s)}\Big);
$$
\item Set $s=s+1$;
\end{enumerate}
\end{enumerate}
\end{algo}
Thus, after an appropriate burn-in period of, say $m$ iterations, the
sequence $(k^{(s)}\,\bdeta_{k^{(s)}}^{(s)})_{s=m+1}^{M}$
provides a sample from the posterior distribution $\Pi^{n}(k, \bdeta_{k}| \by_{1:n})$.

An important goal of an extreme value analysis is to predict the probability of future simultaneous exceedances.
A simple way to do so is to use formula \eqref{eq:simu_excess}. This task can be fully performed, within the Bayesian paradigm, through a Monte Carlo estimate of the posterior predictive distribution, i.e.
\begin{equation}\label{lab:predictive}
\prob(Y_1>y^*_1, Y_2>y^*_2|\by_{1:n})=\int_{\btheta \in \bTheta}
\prob(Y_1>y^*_1, Y_2>y^*_2|\btheta) \, \Pi^n(\btheta|\by_{1:n}) \, \diff \btheta,
\end{equation}
where $y_1^*, y_2^*>0$ are unobserved thresholds. For each element of the posterior sample, applying expressions
\eqref{eq:simu_excess} 
 and \eqref{eq:bpoly_angdist}, 
we have that
\begin{eqnarray*}\label{eq_simu_exc_bern}
\prob(Y_1>y^*_1, Y_2>y^*_2|\btheta)&=&\frac{1}{k}\sum_{j=0}^{k-2}(\eta_{j+1}-\eta_{j}) \\
&\times& \Bigg(\frac{(j+1) \, \text{B}\big(\, y^{*}_1/(y^{*}_1+y^{*}_2) \; | \; j+2,k-j-1\big)}{y^{*}_1} \\
&+& \frac{(k-j-1) \, \text{B}\big( \, y^{*}_2/(y^{*}_1+y^{*}_2) \; |\; k-j,j+1\big)}{y^{*}_2}
\Bigg),
\end{eqnarray*}
where $B(x|a,b)$, for $x\in\spa$, denotes the cumulative distribution function of a Beta random variable with shape parameters $a,b>0$. Therefore, an estimate can be obtained by averaging these quantities over the complete posterior sample.

The efficacy of our proposed model and inference methodology is numerically illustrated in the next section.
%

%
\section{Numerical Examples}\label{sec:simu}
%

We illustrate the performance and flexibility of our methodology 
through a simulation study in which the extremal dependence of some well-known parametric models is inferred.
In particular, we consider the symmetric logistic (SL) model (\citeNP[p. 146]{coles01}), the asymmetric logistic (AL) model \cite{tawn90}, the H\"{u}sler-Reiss (HR) model \cite{husler+r89} and  the Extremal-$t$ (ET) model  \cite{nikoloulopoulos2009}.

For each model, a sample of $n=100$ bivariate observations with common unit Fr\'{e}chet marginal distributions is simulated. 
Using such datasets, MCMC posterior samples of the angular measure and the Pickands dependence function are simulated via Algorithm \ref{algo:metro_etas} and compared 
with the theoretical functions (see Figures \ref{fig:AL} and \ref{fig:SL-HR-ET}).
After a burn-in period of $m = 400$ thousand iterations, $100$ thousand samples are considered.
Figure \ref{fig:AL} displays the results obtained using the AL model with a mild dependence structure. In particular,
the dependence parameter is $\alpha=0.6$ and the asymmetry parameters are $(\tau_{1}, \tau_{2}) = (0.3, 0.8)$.
The four columns report the results attained using different prior distributions for the polynomial degree $k$, when
modeling the distribution function $H_{k-1}$. Precisely, from left to right, a Poisson distribution with mean 
$\kappa_P = 7 $ and a negative binomial with parameters ($\kappa_{NB} = 0.57, \sigma^2 = 0.73$),  
($\kappa_{NB} = 12.40, \sigma^2 = 23.66$), ($\kappa_{NB} = 3.2, \sigma^2 = 4.48$) have been considered.
The third row shows the prior and posterior distributions for $k$, 
in green and red, respectively.  The posterior median values are equal to 9, 3, 13 and 5, respectively for the four cases, 
from left to right.
In the fourth row the prior (green line) and posterior (red line) distributions for the atom $p_{0}$, in addition to its true value $p_0=(1-\tau_2)/2=0.35$ (black dashed line) are reported.
For all the cases, we see that most of the mass of the posterior distribution is concentrated close to the true value.
The corresponding median values of the posterior distributions for $p_0$ are 0.351, 0.342, 0.358 and 0.349, from left to right. 
For the atom $p_1=(1-\tau_1)/2=0.10$
we obtain the median values 0.149, 0.155, 0.165 and 0.139. Then, we can conclude that the information about the
point masses at the edges of the unit interval is well reproduced. The first and second rows report
the point-wise mean (red line) and the point-wise $95\%$ credibility bands (in grey) computed through the posterior samples of the angular density and the Pickands dependence function, respectively. 
The credibility bands are the point-wise $0.05$- and $0.95$-quantiles of the posterior samples. 
The solid black lines are the true functions. 
 In the first row, the true point masses on the edges are represented by black dots and the means computed from the posterior distributions are represented by red dots.
The grey points are $95\%$ upper and lower limits of the credibility intervals for the point masses. 
The true functions (angular density and Pickands) 
and the point masses fall within the point-wise $95\%$ credibility bands in most of the cases, pointing out that our
inferential method captures the dependence structure quite well. 
In the four cases, the results are quite similar; only in the  third column (from the left), 
the $95\%$ credibility bands do not include the true functions in a few points. 
So, it seems that our method is not too sensitive to the prior distribution. 
The fifth row reports the Monte Carlo predictive probabilities (red lines) of future simultaneous exceedances \eqref{lab:predictive} for pairs of unobserved thresholds $(y^{*}_{1},y^{*}_{2})$ ranging between $10$ and $100$.
The black lines are the true probabilities. In the second and fourth cases (from the left)
 the estimates are very accurate, while they are less so for the first and the third.

Figure \ref{fig:SL-HR-ET} reports (from left to right) the results obtained for data generated from the SL model with mild and weak dependence structures (denoted by SLm and SLw) given by dependence parameter values $\alpha=0.45$ and $0.85$; the HL and the ET models with mild dependence given by the dependence parameters $\lambda = 1.2$ and ($\omega=0.8$, $\nu=2$), respectively. The format of the graphs is the same as that of the previous figure. The results are obtained using the same prior distribution for the polynomial degree $k$, i.e. a Poisson distribution with mean $\kappa_P = 7$. 
Other prior settings can be considered (skipped here for brevity) and, as the previous study shows, the results do not change significantly.
In the four cases, the posterior median values for $k$ are 6, 9, 7 and 8. The only model that
includes point masses on the edges is the ET, 
corresponding to $p_0=p_1=T_{\nu+1}(-\omega\{(\nu+1)/(1-\omega^2)\}^{1/2})=0.104$, where $T_{\nu+1}(\cdot)$ denotes a $t$ distribution with $\nu+1$ degrees of freedom. The medians of the posterior distributions for the point masses are (0.018, 0.041), (0.235, 0.151), (0.047, 0.045) and (0.041, 0.097), respectively for the four cases. We see, from the 
first and second row-panels, that the posterior distributions adequately capture the different extremal dependence forms. 
Also, the plots in the last row show 
rather accurate predictions of the probabilities of joint exceedances, outlining the good performance of our inferential method.

Finally, going beyond visual checks, we measure the accuracy of our proposed method. To do so, we focus on
the Pickands dependence function and we compute, for each element of the posterior
MCMC sample, the integrated squared error: 
\begin{equation*}
\label{eq:mise}
\mbox{ISE}(A^{(s)},A) = \int_0^1 \Big(A^{(s)}(t) - A(t)\Big)^2 \diff t, 
\end{equation*}
where $A$ is the true Pickands dependence function and $A^{(s)}$, $s=1,\ldots,m$, is a Pickands dependence function sampled from the posterior. Table \ref{tab:MISE} reports, for different sample sizes (first column) and for 
each of the four models considered in Figures \ref{fig:AL} and \ref{fig:SL-HR-ET} (second column), 
the Monte Carlo posterior mean of the ISE (third column). Between parenthesis the $0.05$- and 
$0.95$- quantiles of the posterior distribution for the ISE are reported. For comparison purposes,
the fourth and fifth columns report similar estimates obtained using the projection method
discussed in \shortciteN{marcon+p+n+m15} and focusing on the multivariate madogram (MD) and Cap\'{e}ra\`{a}-Foug\`{e}res-Genest (CFG, \shortciteNP{cap+f+g97}) estimators as pilot estimates, see \shortciteN{marcon+p+n+m15} for details. In particular, for each dataset, $500$ bootstrap replicates are produced and for each of these
the ISE is computed, where in this case $A^{(s)}$ is the estimated Pickands dependence function obtained
with the projection method. 
In the table, the mean and the $0.05$- and 
$0.95$- quantiles (in parenthesis) of the ISE computed over the $500$ bootstrap replicates, are reported.
Results in Table \ref{tab:MISE} show the slightly better performance of our proposed method with respect to the competitors, for the several examples considered. This supports our new proposal.

We have also compared our inferential approach with other proposals as for instance that in \shortciteN{einmahl2008} (see also, \shortciteNP{kluppelberg2007}, \citeNP{krajina2012}).
They proposed a parametric method for estimating the tail of a bivariate distribution that is in the domain of attraction of a bivariate extreme value distribution \cite[Ch. 6]{dehaan+f06}. Instead, we directly model the extremal dependence of bivariate extreme value distributions. Thus, care must be taken when interpreting the results, which for brevity are not presented here. For specific parametric families of dependence models, their parametric method outperforms our nonparametric proposal. However, the integrated squared error (used for comparison) is of the same order in both techniques, suggesting that our model-free proposal is equally appealing, in addition to providing wide applicability.

In conclusion, we stress that the computational cost of running our proposed Bayesian model is moderately low. For example, to run $M = 500$ thousand iterations of the MCMC algorithm, it takes only 114.03 seconds, with an intel Core i7 processor at 2.2 GHz. The code for the model fitting will soon be available with
the {\tt R}-package {\tt ExtremalDep}. The data simulation was performed using the {\tt R}-package {\tt EVD} \cite{stephenson2004}.

%
%
\begin{table}[h]
\begin{center}
\caption{Mean, $95\%$ credibility intervals (Bayesian method) and $95\%$ bootstrap confidence intervals (projection method) of the ISE for the models in Figures \ref{fig:AL} and \ref{fig:SL-HR-ET}, for increasing sample sizes. }
{\footnotesize \begin{tabular}{cccc}
\toprule
Model	&	 \multicolumn{3}{c}{Inferential methods}\\	
\midrule
&\multicolumn{3}{c}{Sample size $25$}\\							
 	&	Bayesian 	& Projection-MD &	Projection-CFG \\
\midrule 
AL & $ 2.35\times 10^{-3} $ & $ 5.10\times 10^{-3} $ & $ 1.13\times 10^{-2} $ \\
& ($3.53\times 10^{-4} ; 5.65\times 10^{-3}$) & ($ 1.81\times 10^{-4} ; 1.90\times 10^{-2} $) & ($8.02 \times 10^{-4} ; 2.45\times 10^{-2} $) \\
 $\text{SLm}$ & $ 7.64\times 10^{-3} $ & $ 6.63\times 10^{-3} $ & $ 1.47\times 10^{-3} $ \\
& ($ 8.33\times 10^{-4} ; 2.09\times 10^{-2}$) & ($ 8.57\times 10^{-5} ; 3.01\times 10^{-2} $) & ($ 6.15\times 10^{-5} ; 6.98\times 10^{-3} $) \\
 $\text{SLw}$ & $ 1.75\times 10^{-3} $ & $ 3.81\times 10^{-3} $ & $ 4.36\times 10^{-3} $ \\
& ($1.23\times 10^{-4} ; 4.21\times 10^{-3}$) & ($ 3.26\times 10^{-4} ; 6.10\times 10^{-3} $) & ($ 3.95\times 10^{-4} ; 1.31\times 10^{-2} $) \\
 HR & $ 8.75\times 10^{-3} $ & $ 4.58\times 10^{-3} $ & $ 6.75\times 10^{-3} $ \\ 
& ($ 4.95\times 10^{-4} ; 1.75\times 10^{-2}$) & ($ 3.10\times 10^{-4} ; 9.94\times 10^{-3} $) & ($ 1.51 \times 10^{-3} ; 9.92\times 10^{-3} $) \\
 ET &  $ 3.43\times 10^{-2} $ & $ 7.00\times 10^{-2} $ & $ 6.55\times 10^{-2} $ \\
& ($2.35\times 10^{-2} ; 5.18\times 10^{-2}$) & ($ 6.17 \times 10^{-2} ; 8.63\times 10^{-2} $) & ($ 6.18\times 10^{-2} ; 7.31\times 10^{-2} $) \\
\midrule
&\multicolumn{3}{c}{Sample size $50$}\\							
&	Bayesian 	& Projection-MD &	Projection-CFG \\
\midrule 
AL & $ 1.23\times 10^{-3} $ & $ 2.04\times 10^{-3} $ & $ 1.96\times 10^{-3} $ \\
& ($4.73\times 10^{-5} ; 4.09\times 10^{-3}$) & ($ 1.10\times 10^{-4} ;6.48\times 10^{-3} $) & ($8.67 \times 10^{-5} ;6.67\times 10^{-3} $) \\
 $\text{SLm}$ & $ 1.76\times 10^{-3} $ & $ 6.52\times 10^{-4} $ & $ 4.17\times 10^{-4} $ \\
& ($1.16\times 10^{-4} ; 5.04\times 10^{-3}$) & ($ 2.53\times 10^{-5} ;2.35\times 10^{-3} $) & ($1.87 \times 10^{-5} ; 1.18\times 10^{-3} $) \\
 $\text{SLw}$ & $ 1.47\times 10^{-3} $ & $ 2.14\times 10^{-3} $ & $ 2.33\times 10^{-3} $ \\
& ($9.18\times 10^{-5} ; 3.89\times 10^{-3}$) & ($ 3.74\times 10^{-4} ;5.59\times 10^{-3} $) & ($2.39 \times 10^{-4} ; 7.08\times 10^{-3} $) \\
 HR & $ 8.87\times 10^{-4} $ &	 $ 2.71\times 10^{-3} $ & $ 4.38\times 10^{-3} $ \\
& ($4.53\times 10^{-5} ; 3.26\times 10^{-3}$) & ($ 2.47\times 10^{-4} ; 6.82\times 10^{-3} $) & ($9.63 \times 10^{-4} ; 8.29\times 10^{-3} $) \\
 ET& $ 3.20\times 10^{-2} $ & $ 7.46\times 10^{-2} $ &	 $ 7.08\times 10^{-2} $ \\
& ($2.42\times 10^{-2} ;4.51\times 10^{-2}$) & ($ 6.68\times 10^{-2} ;8.52\times 10^{-2} $) & ($6.53 \times 10^{-2} ; 7.69\times 10^{-2} $) \\
\midrule
&\multicolumn{3}{c}{Sample size $100$}\\							
&	Bayesian 	& Projection-MD &	Projection-CFG \\
\midrule 
AL &$ 5.71\times 10^{-4} $ &	 $ 9.48\times 10^{-4} $ &	 $ 6.51\times 10^{-4} $ \\
& ($1.60\times 10^{-5} ;2.02\times 10^{-3}$) & ($ 7.10\times 10^{-5} ;6.47\times 10^{-3} $) & ($2.98 \times 10^{-5} ;2.30\times 10^{-3} $) \\
 $\text{SLm}$ &$ 3.58\times 10^{-4} $ & 	 $ 1.85\times 10^{-4} $ &	 $ 1.91\times 10^{-4} $ \\
& ($7.67\times 10^{-6} ;1.13\times 10^{-3}$) & ($ 1.53\times 10^{-5} ;2.85\times 10^{-4} $) & ($2.10 \times 10^{-5} ;2.84\times 10^{-4} $) \\
 $\text{SLw}$ &$ 8.44\times 10^{-4} $ & 	 $ 1.21\times 10^{-3} $ &	 $ 1.17\times 10^{-3} $ \\
& ($4.77\times 10^{-5} ;2.74\times 10^{-3}$) & ($ 9.67\times 10^{-5} ;3.88\times 10^{-3} $) & ($1.23 \times 10^{-4} ;4.02\times 10^{-3} $) \\
 HR &$ 5.61\times 10^{-4} $ & 	 $ 2.16\times 10^{-3} $ &	 $ 2.37\times 10^{-3} $ \\
& ($3.89\times 10^{-5} ;1.67\times 10^{-3}$) & ($ 2.32\times 10^{-4} ;4.71\times 10^{-3} $) & ($5.38 \times 10^{-4} ;4.24\times 10^{-3} $) \\
 ET &$ 2.49\times 10^{-2} $ & 	 $ 6.66\times 10^{-2} $ &	 $ 6.68\times 10^{-2} $ \\
& ($2.14\times 10^{-2} ;2.99\times 10^{-2}$) & ($ 6.49\times 10^{-2} ;7.23\times 10^{-2} $) & ($6.49 \times 10^{-2} ;7.09\times 10^{-2} $) \\
\midrule
&\multicolumn{3}{c}{Sample size $200$}\\							
&	Bayesian 	& Projection-MD &	Projection-CFG \\
\midrule 
AL &$ 3.76\times 10^{-4} $ &	 $ 6.09\times 10^{-4} $ &	 $ 4.95\times 10^{-4} $ \\
& ($1.87\times 10^{-5} ;1.22\times 10^{-3}$) & ($ 4.50\times 10^{-5} ;1.92\times 10^{-3} $) & ($3.31 \times 10^{-5} ;1.63\times 10^{-3} $) \\
 $\text{SLm}$ &$ 5.62\times 10^{-5} $ &	 $ 4.52\times 10^{-4} $ &	 $ 4.84\times 10^{-4} $ \\
& ($6.45\times 10^{-6} ;1.50\times 10^{-4}$) & ($ 4.69\times 10^{-5} ;1.03\times 10^{-3} $) & ($1.01 \times 10^{-4} ;9.65\times 10^{-4} $) \\
 $\text{SLw}$ &$ 5.16\times 10^{-4} $ &	 $ 8.10\times 10^{-4} $ &	 $ 1.19\times 10^{-3} $ \\
& ($2.87\times 10^{-5} ;1.72\times 10^{-3}$) & ($ 4.59\times 10^{-5} ;2.54\times 10^{-3} $) & ($6.41 \times 10^{-5} ;3.39\times 10^{-3} $) \\
 HR &$ 2.53\times 10^{-4} $ &	 $ 3.91\times 10^{-4} $ &	 $ 3.62\times 10^{-4} $ \\
& ($1.73\times 10^{-5} ;8.55\times 10^{-4}$) & ($ 2.22\times 10^{-5} ;1.20\times 10^{-3} $) & ($2.59 \times 10^{-5} ;1.09\times 10^{-3} $) \\
 ET &$ 2.28\times 10^{-2} $ & $ 6.25\times 10^{-2} $ &	 $ 6.16\times 10^{-2} $ \\
& ($2.09\times 10^{-2} ;2.56\times 10^{-2}$) & ($6.10 \times 10^{-2} ;6.93\times 10^{-2} $) & ($6.15 \times 10^{-2} ;6.93\times 10^{-2} $) \\
\bottomrule													
\end{tabular}}
\label{tab:MISE}
\end{center}
\end{table}
%

\begin{figure}
\centering
\includegraphics[width=.24\textwidth, page=4]{AL_casoA}
\includegraphics[width=.24\textwidth, page=4]{AL_casoB}
\includegraphics[width=.24\textwidth, page=4]{AL_casoC}
\includegraphics[width=.24\textwidth, page=4]{AL_casoD}\\
\includegraphics[width=.24\textwidth, page=3]{AL_casoA}
\includegraphics[width=.24\textwidth, page=3]{AL_casoB}
\includegraphics[width=.24\textwidth, page=3]{AL_casoC}
\includegraphics[width=.24\textwidth, page=3]{AL_casoD}\\
\includegraphics[width=.24\textwidth, page=1]{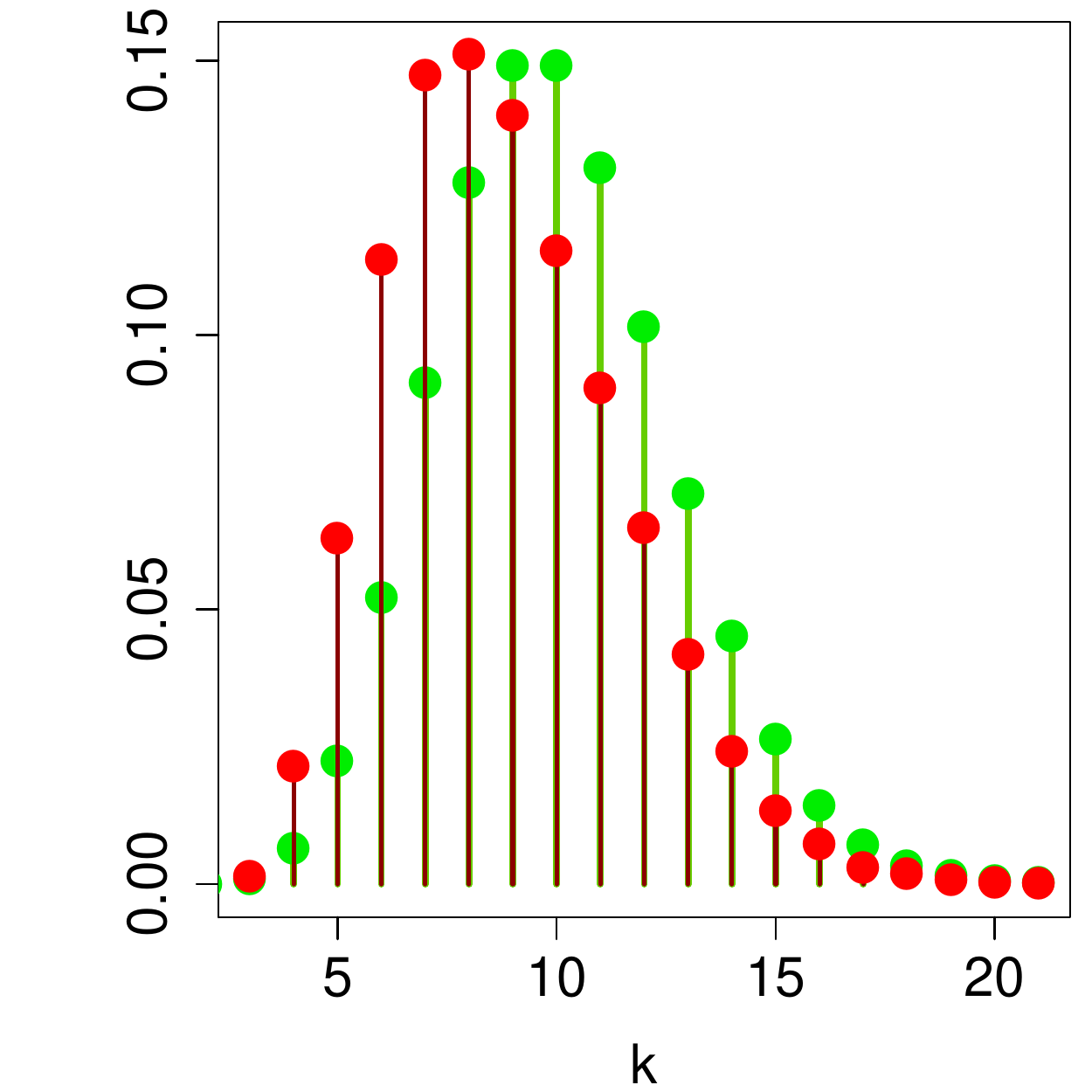}
\includegraphics[width=.24\textwidth, page=1]{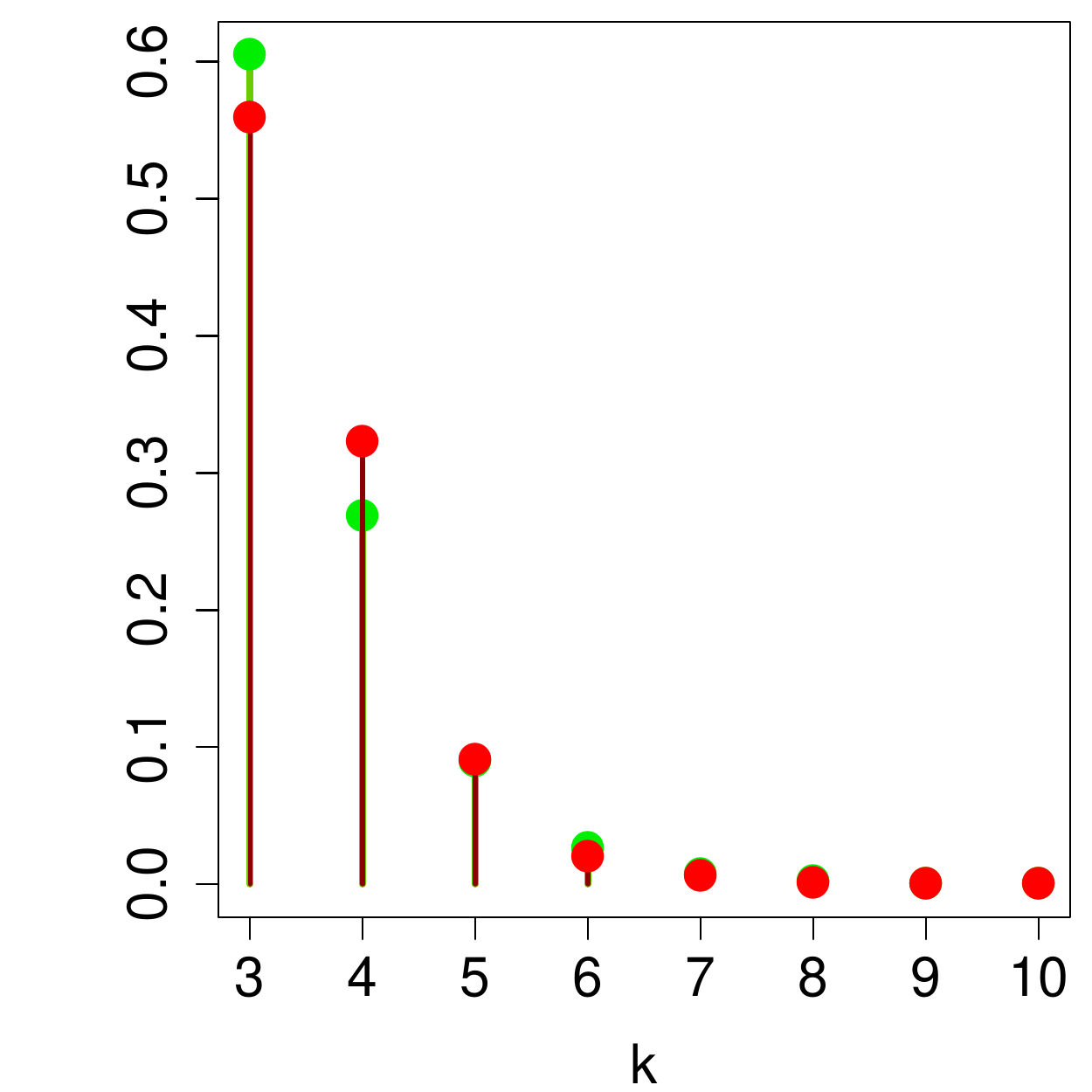}
\includegraphics[width=.24\textwidth, page=1]{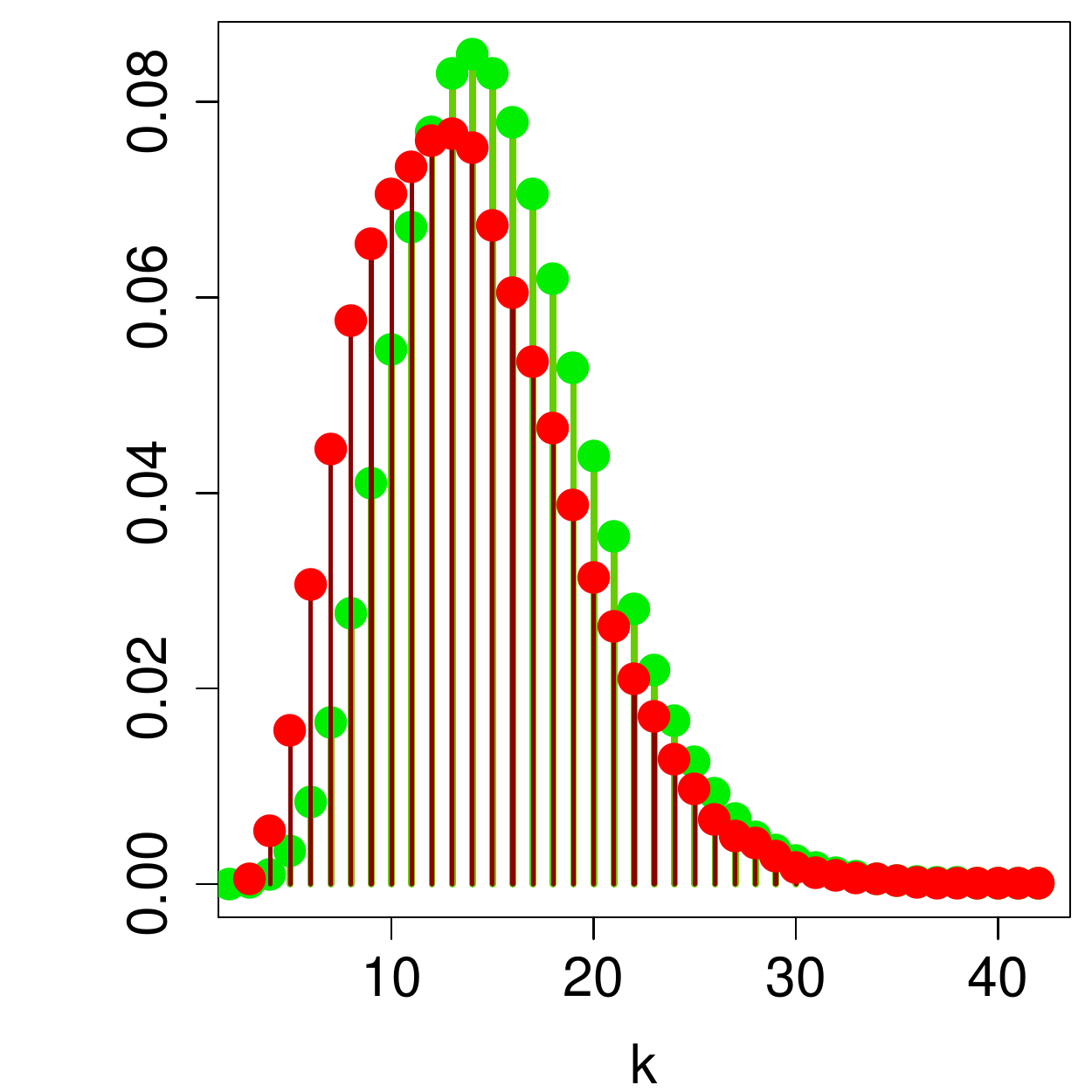}
\includegraphics[width=.24\textwidth, page=1]{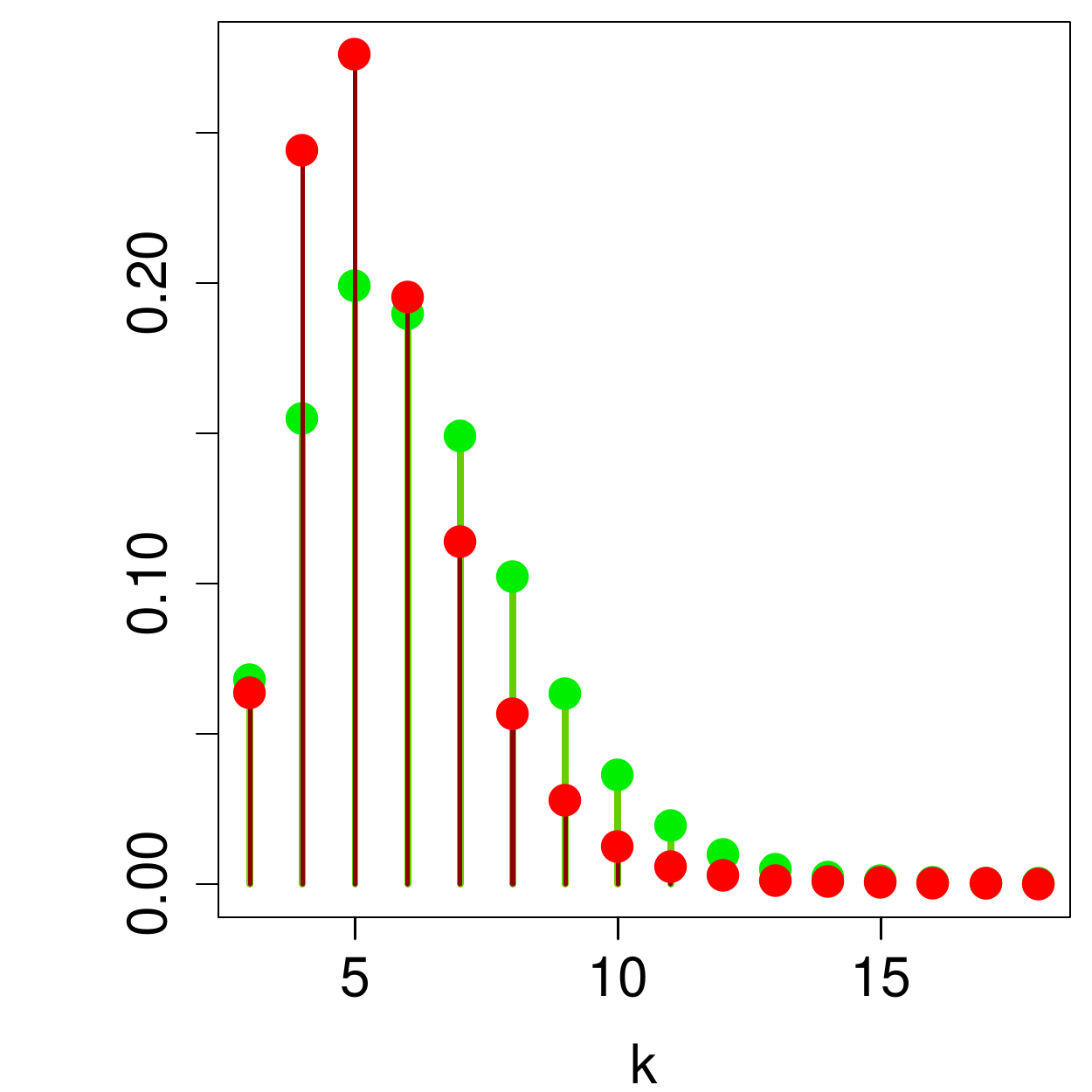}\\
\includegraphics[width=.24\textwidth, page=2]{AL_casoA}
\includegraphics[width=.24\textwidth, page=2]{AL_casoB}
\includegraphics[width=.24\textwidth, page=2]{AL_casoC}
\includegraphics[width=.24\textwidth, page=2]{AL_casoD}\\
\includegraphics[width=.24\textwidth,page=1]{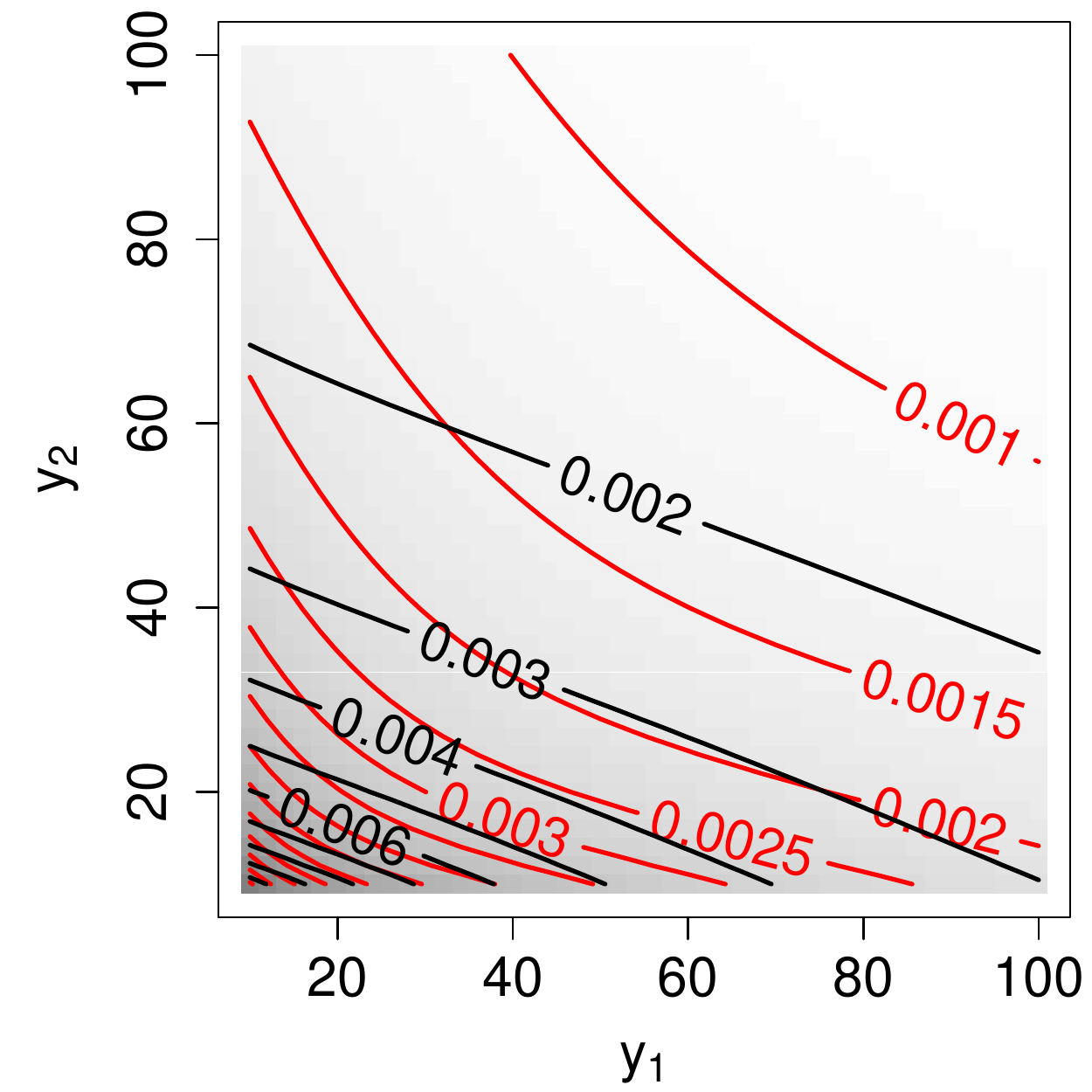}
\includegraphics[width=.24\textwidth,page=2]{returns_AL_casoA-D}
\includegraphics[width=.24\textwidth,page=3]{returns_AL_casoA-D}
\includegraphics[width=.24\textwidth,page=4]{returns_AL_casoA-D}\\
\caption[Summary of the Bayesian nonparametric fitting of the extremal dependence]
{\small Summary of the Bayesian nonparametric fitting of the extremal dependence. The true model is the Asymmetric Logistic model. Different prior distributions for the polynomial's degree $k$ are considered from left to right.}
\label{fig:AL}
\end{figure}

\begin{figure}
\centering
\includegraphics[width=.24\textwidth, page=4]{SL1_pois7}
\includegraphics[width=.24\textwidth, page=4]{SL2_pois7}
\includegraphics[width=.24\textwidth, page=4]{HR_pois7}
\includegraphics[width=.24\textwidth, page=4]{ET_pois7}\\
\includegraphics[width=.24\textwidth, page=3]{SL1_pois7}
\includegraphics[width=.24\textwidth, page=3]{SL2_pois7}
\includegraphics[width=.24\textwidth, page=3]{HR_pois7}
\includegraphics[width=.24\textwidth, page=3]{ET_pois7}\\
\includegraphics[width=.24\textwidth, page=1]{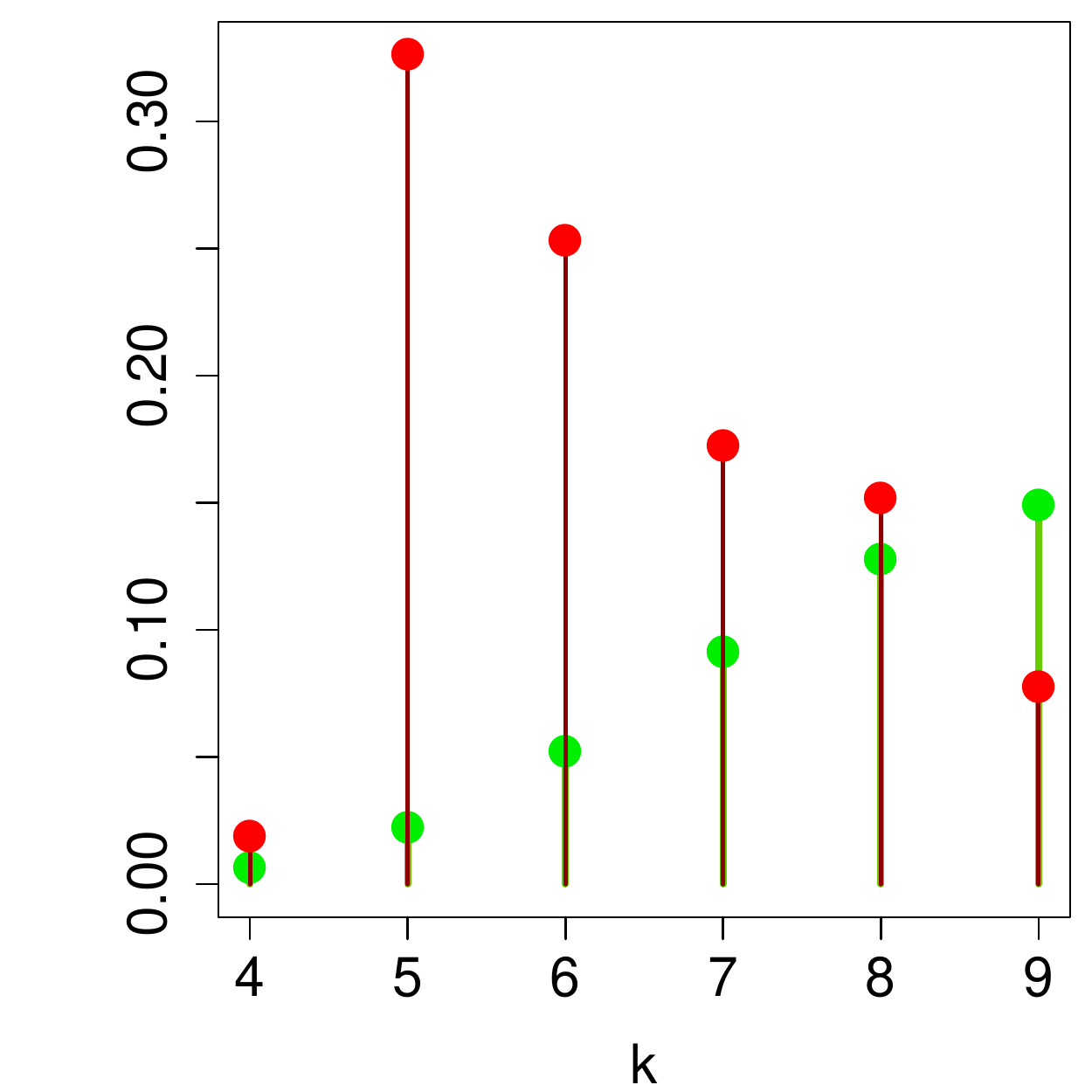}
\includegraphics[width=.24\textwidth, page=1]{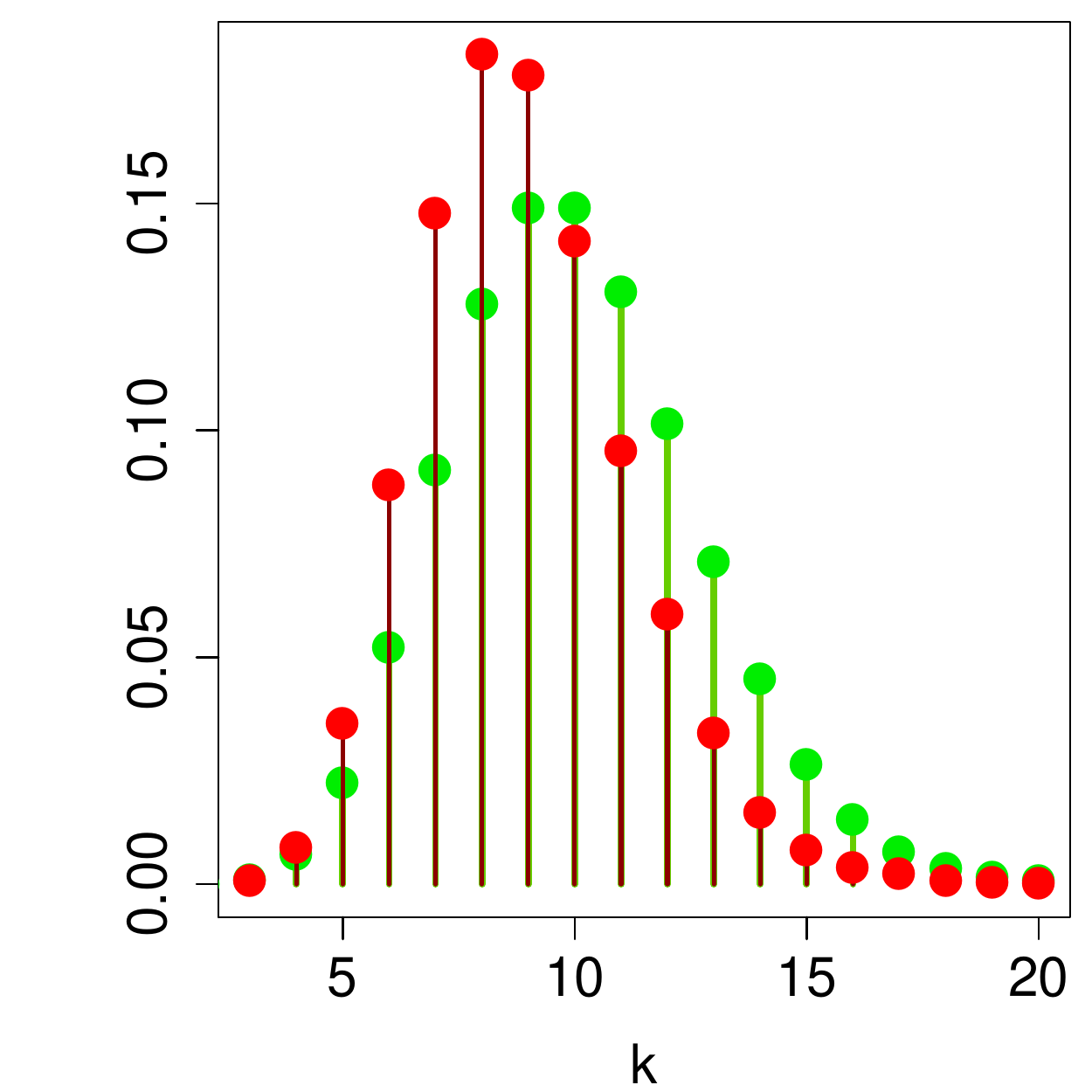}
\includegraphics[width=.24\textwidth, page=1]{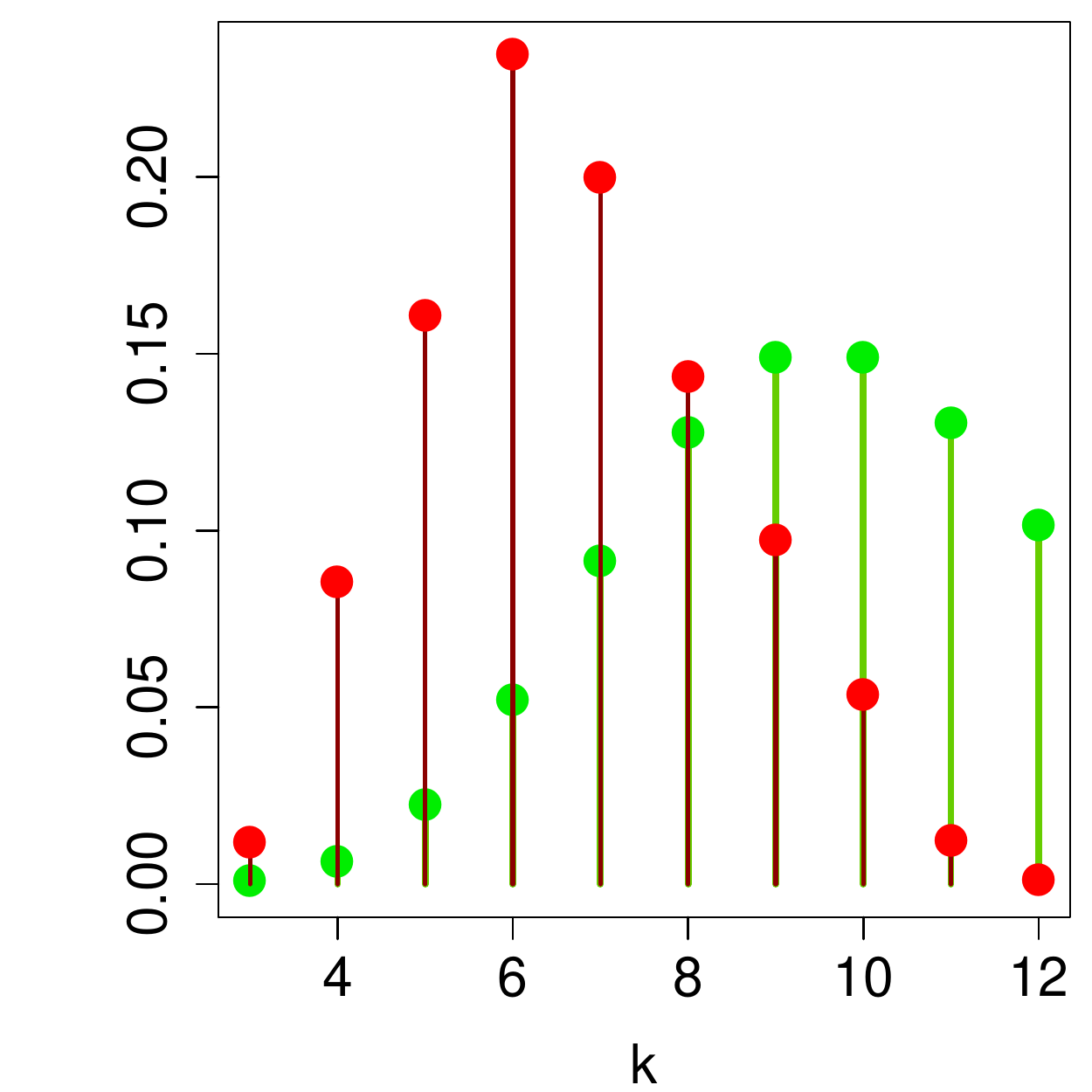}
\includegraphics[width=.24\textwidth, page=1]{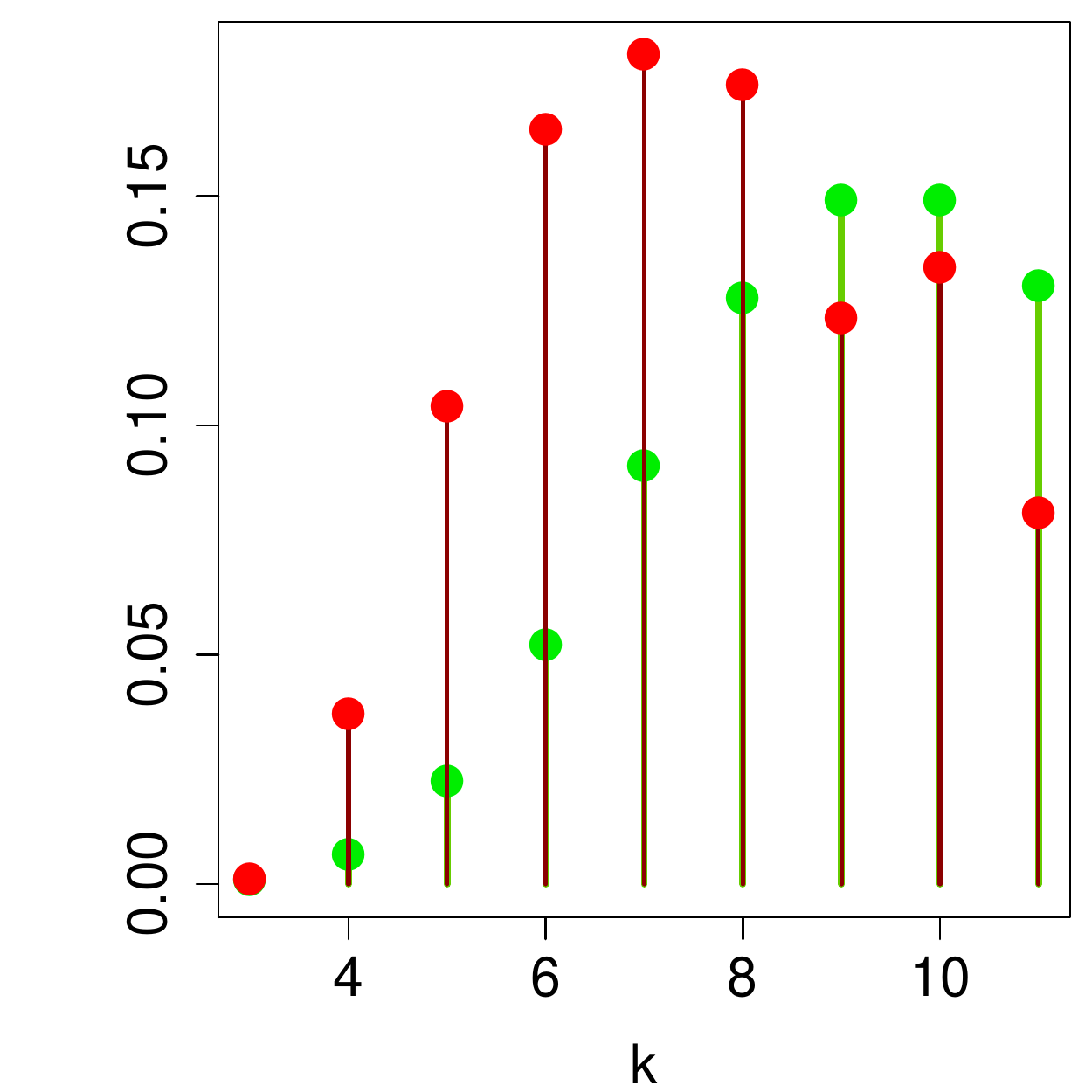}\\
\includegraphics[width=.24\textwidth, page=2]{SL1_pois7}
\includegraphics[width=.24\textwidth, page=2]{SL2_pois7}
\includegraphics[width=.24\textwidth, page=2]{HR_pois7}
\includegraphics[width=.24\textwidth, page=2]{ET_pois7}\\
\includegraphics[width=.24\textwidth]{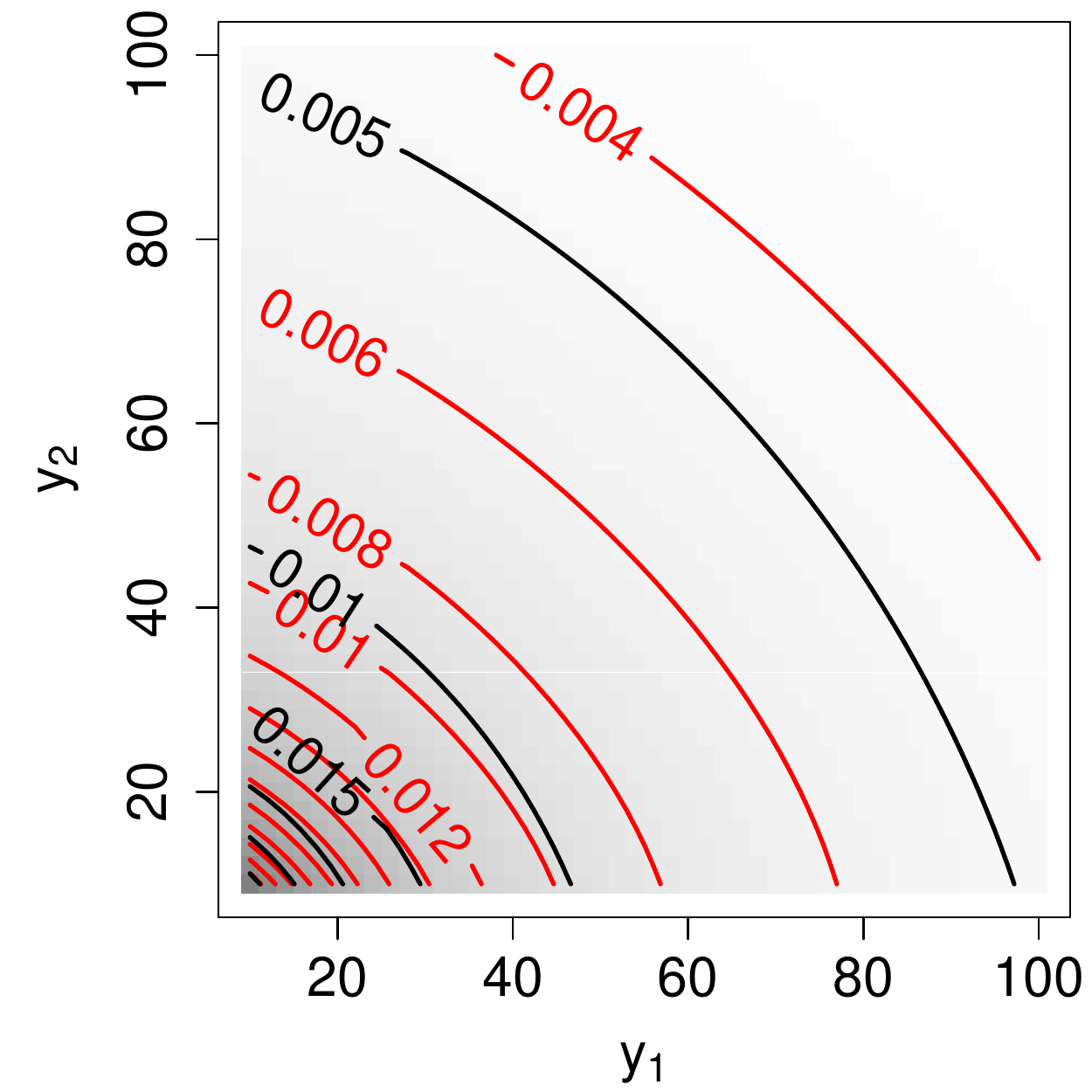}
\includegraphics[width=.24\textwidth]{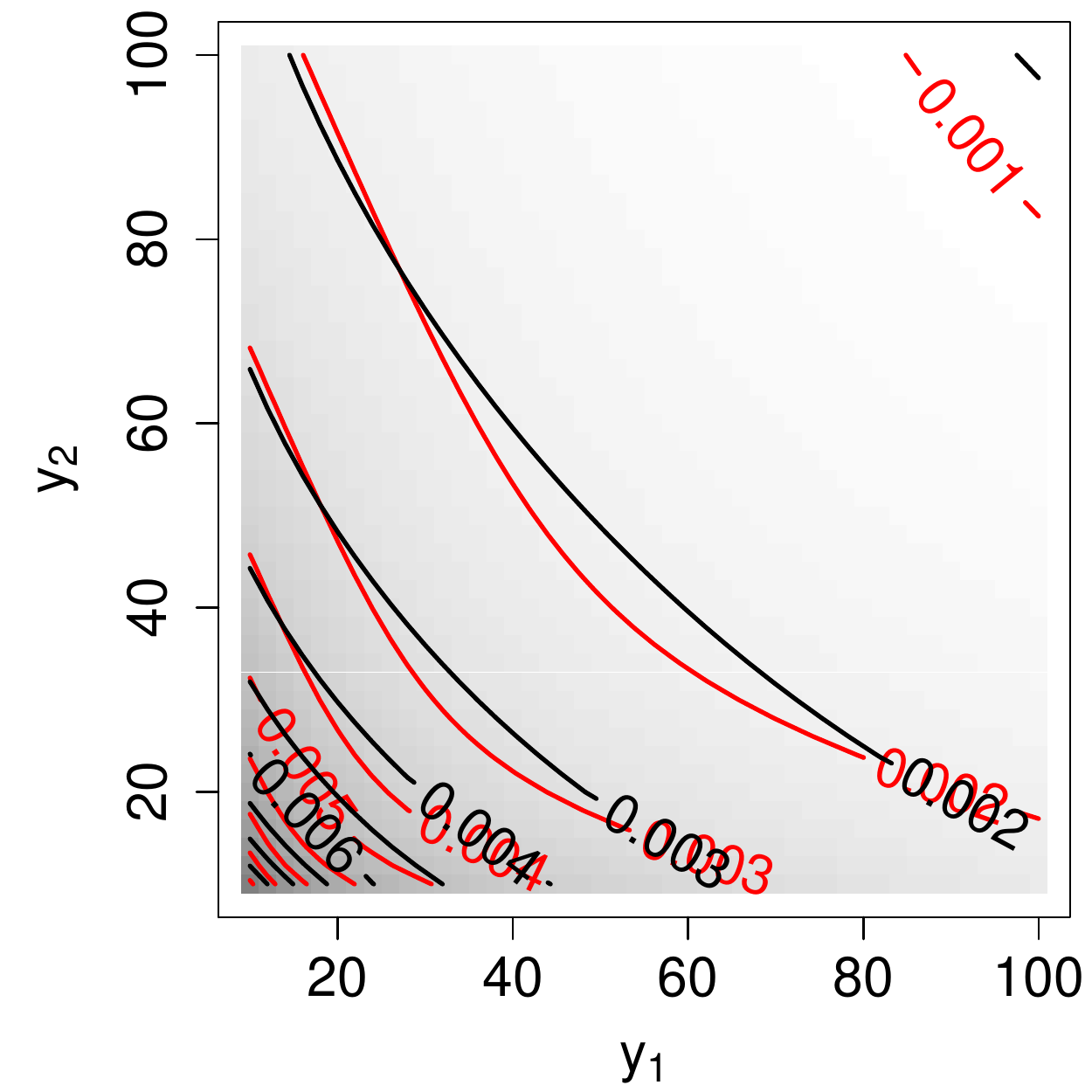}
\includegraphics[width=.24\textwidth]{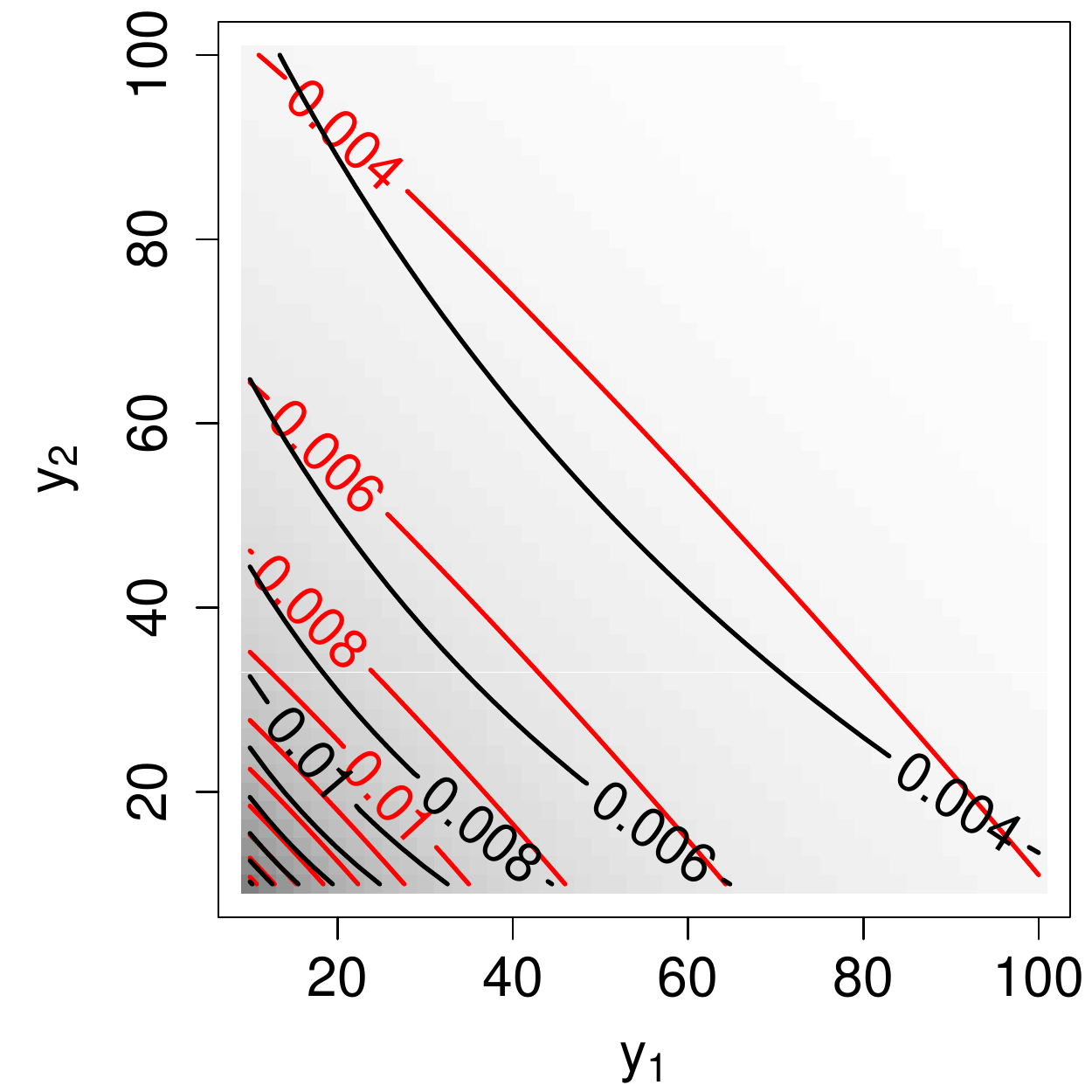}
\includegraphics[width=.24\textwidth]{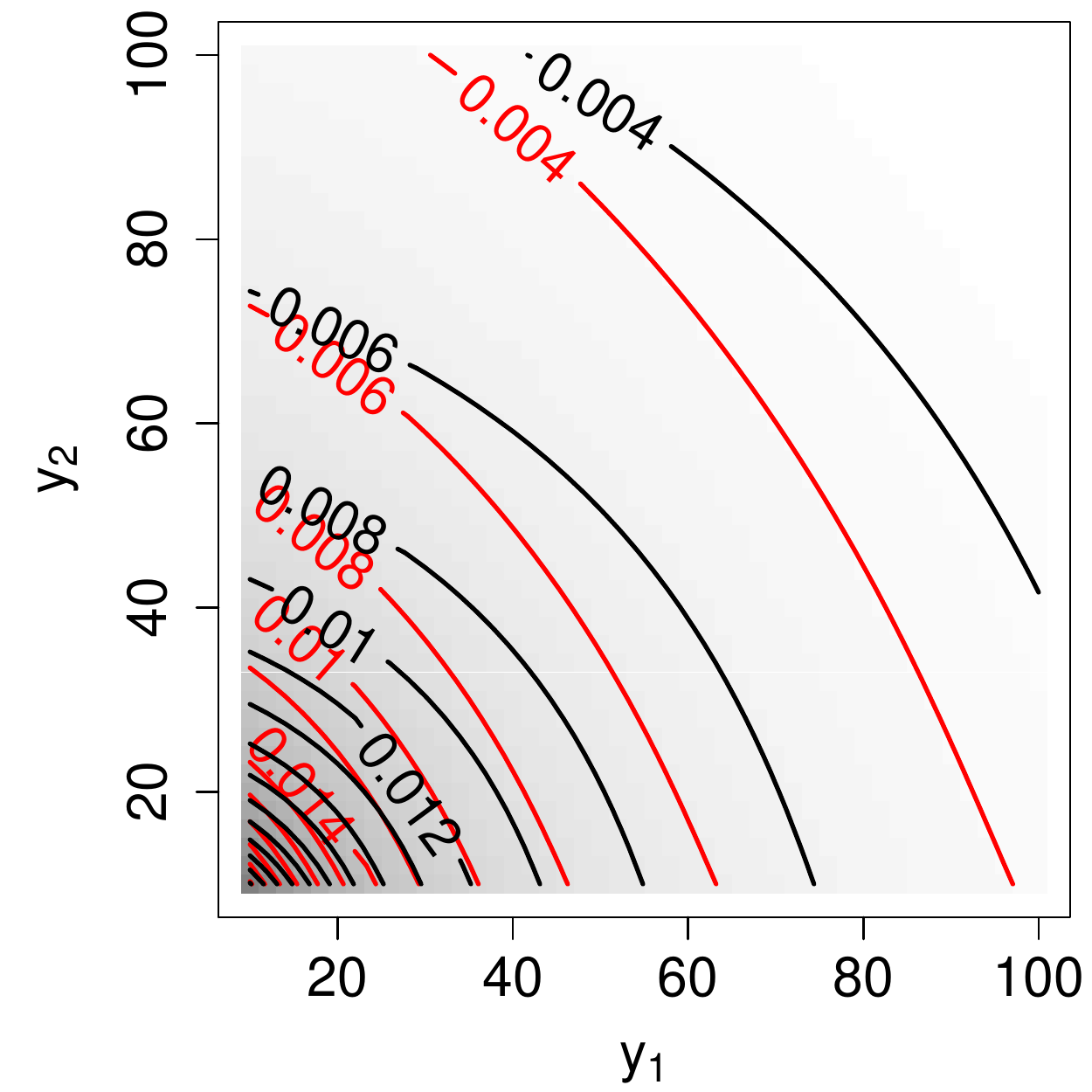}\\
\caption[Bayesian nonparametric estimation of the extremal dependence models]{\small Summary of the Bayesian nonparametric fitting for the extremal dependence model: Symmetric Logistic (mild and weak), 
H\"{u}sler-Reiss and Extremal-$t$.}
\label{fig:SL-HR-ET}
\end{figure}

\section{Analysis of Extreme Log-return Exchange Rates}\label{sec:app}


Predicting exchange rates is one of the most challenging tasks in economics. 
A seminal paper by \citeN{meese1983} showed that predictions of exchange rates based on macroeconomic models are unable to outperform those derived from a random walk.
However, recent literature (e.g. \citeNP{engel2005}) has established a link between exchange rates and fundamental economic principles. The modern asset market approach relies on a supply-and-demand analysis of the exchange rate viewed as the price of domestic assets in terms of foreign assets (\citeNP{madura2014}).
In the short-term, the exchange rate is influenced by a positive interest rate differential, which causes an appreciation of the home currency. In the long-term, a rise in the home country's price level causes the depreciation of its currency, while higher productivity or an increased demand for exports cause the appreciation of the currency (the opposite holds true for an increased demand for imports).
\begin{figure}[b!]
\centering
\includegraphics[width=.85\textwidth]{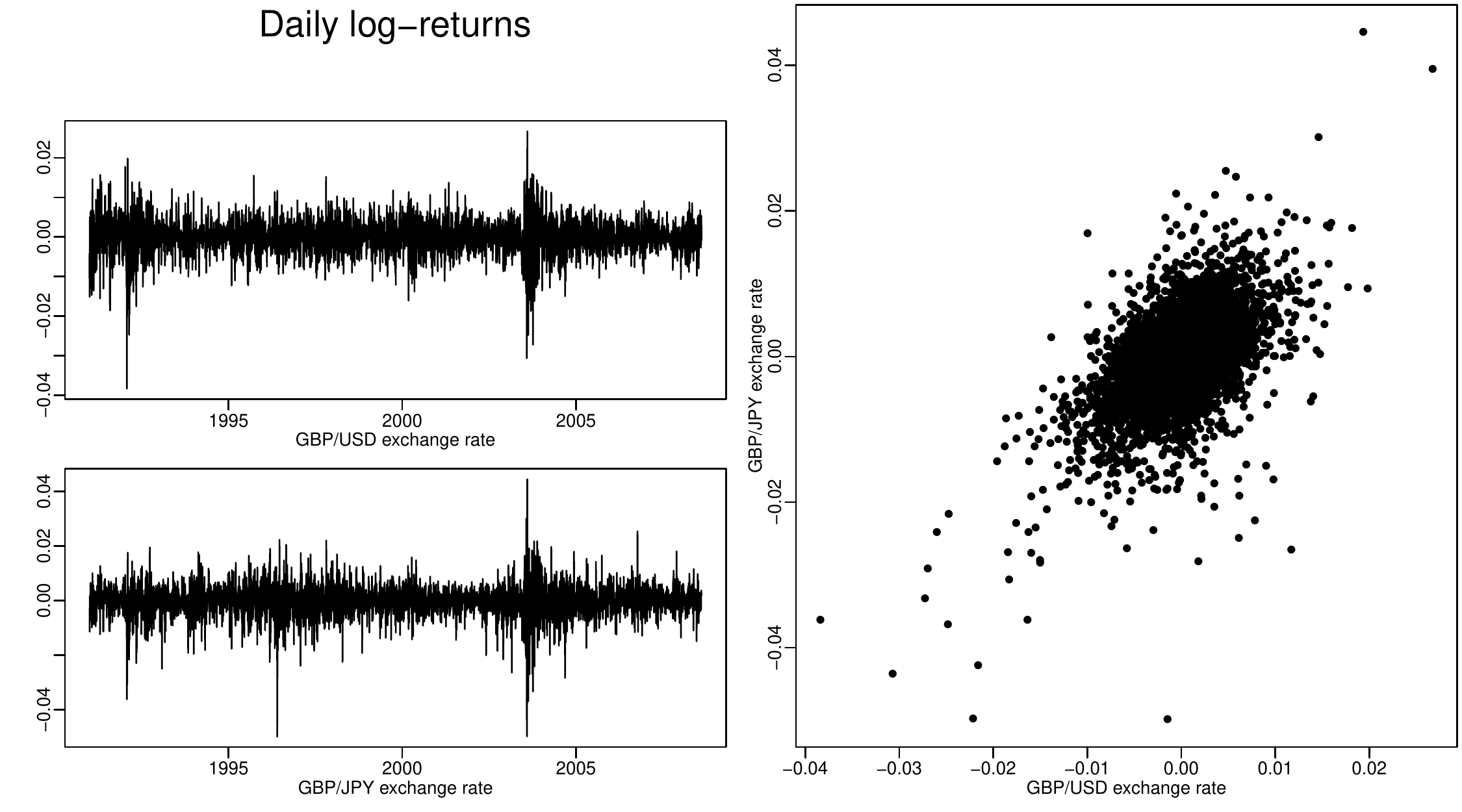}
\caption{Daily log-returns of GBP/USD and GBP/JPY exchange rates.}\label{lab:daily_logret_exc}
\end{figure}

The United States and Japan share some common features, such as the presence of titanic enterprises and a similar monetary policy, so a strong dependence between the exchange rates of the Pound Sterling against the US dollar (GBP/USD) and the Japanese yen (GBP/JPY) is to be expected. In fact, Figure \ref{lab:daily_logret_exc} shows a remarkable relation between the daily log-returns for this pair of exchange rates from March $1991$ to October 2015. Our interest
is in estimating extremely high (or low) joint levels of the exchange rates, thus we focus on monthly-maxima
of log-returns. An inspection of the data shows, for instance, that monthly-maxima often occur on the same day of the month. An adequate quantification of the dependence of the bivariate maxima is crucial for predicting future extremely high exchange
rates of GBP/JPY based on occurrences of extremely high exchange rates of GBP/USD, and vice versa. 
\begin{figure}[t!]
\centering
\includegraphics[width=.85\textwidth]{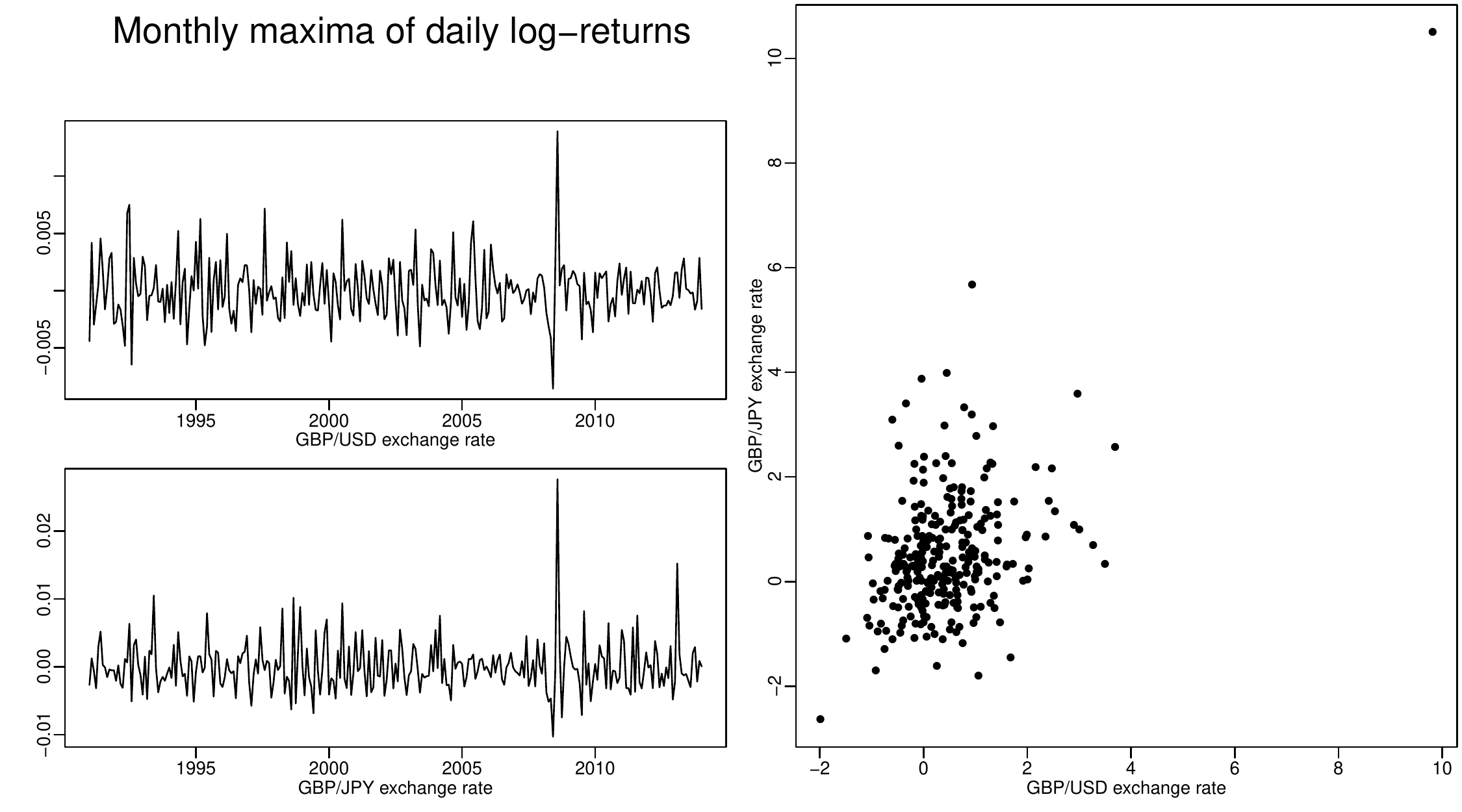}
\caption{Monthly-maxima of log-returns of GBP/USD and GBP/JPY exchange rates.}\label{lab:monthly_logret_exc}
\end{figure}
Figure \ref{lab:monthly_logret_exc} shows that an important degree of extremal dependence persists, even after removing the trend and seasonality from each of the monthly-maxima series. 
Firstly, we estimate the marginal GEV parameters of each series of residuals, by the maximum likelihood method.
The parameter estimates
for GBP/USD and GBP/JPY are $\mu_{1} = 0.0055$, $\sigma_{1}=0.0025$, $\xi_{1}=0.0249$ and $\mu_{2}=0.0068$, $\sigma_{2}=0.0030$, $\xi_{2}=0.1199$, respectively. 
Note that $\xi_{2}$ is higher than $\xi_{1}$. Since the shape parameter drives the heaviness of the tail,
the larger it is, the heavier the tail is, therefore the higher the marginal probability of observing extreme values is for GBP/JPY 
as opposed to GBP/USD.
Secondly, we transform the data to obtain unit Fr\'{e}chet margins, by means of transformation \eqref {eq:gev_frechet}
and using the estimated marginal parameters.
The data transformed in this way can be assumed to be a sample coming approximately from a bivariate max-stable distribution of the type \eqref{eq:bivgev}. 
\begin{figure}[t!]
\centering
\includegraphics[width=.33\textwidth, page=4]{resultsData_pois7}
\includegraphics[width=.33\textwidth, page=3]{resultsData_pois7}
\includegraphics[width=.33\textwidth, page=1]{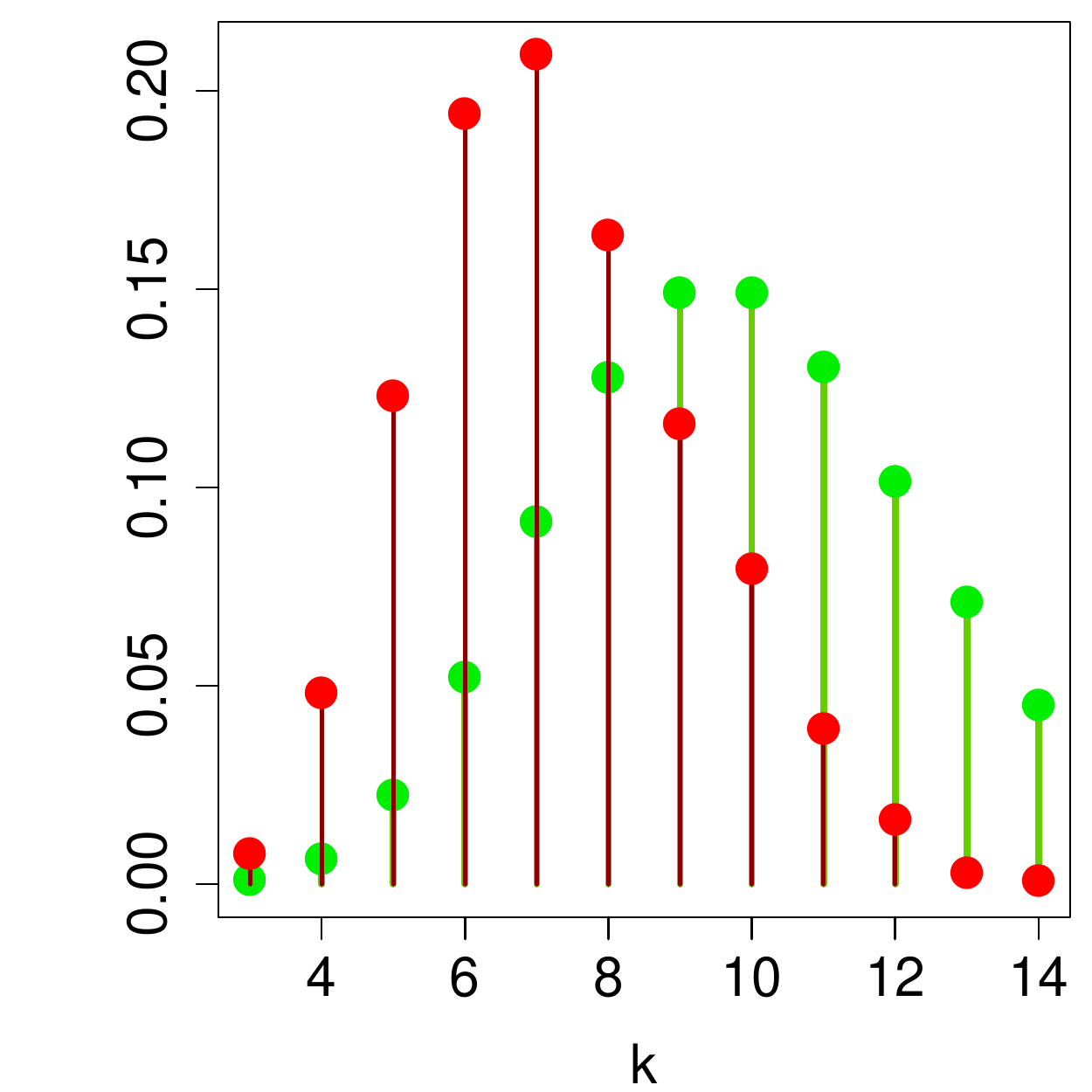}
\includegraphics[width=.33\textwidth]{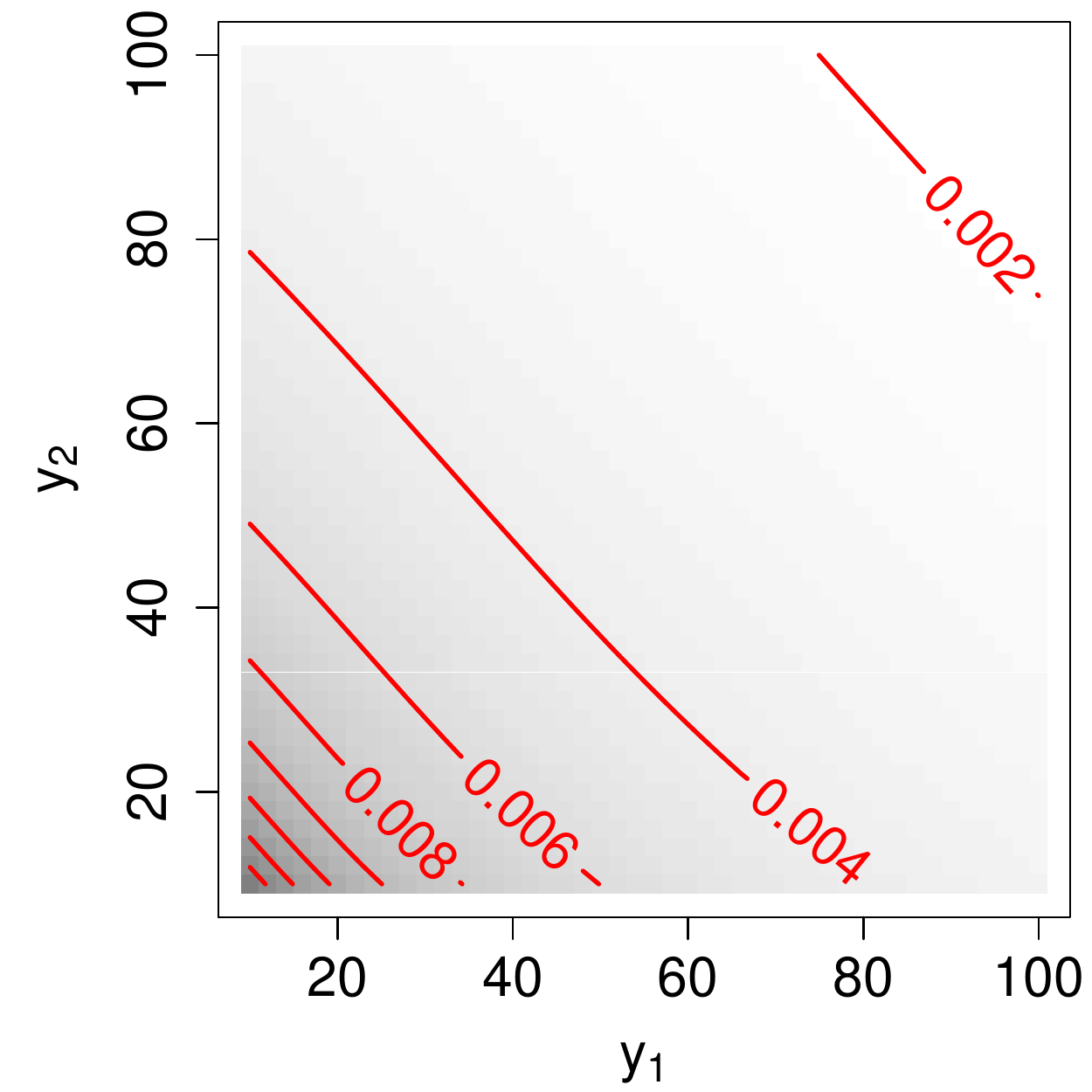}
\caption{Summary of the Bayesian nonparametric fitting of the extremal dependence for the 
monthly-maxima of GBP/USD and GBP/JPY log-returns of exchange rates.}
\label{fig:dataAnalysis}
\end{figure}
The extremal dependence of monthly-maxima of log-returns is then inferred by
using the method described in Section \ref{sec:bayeisan}. 

The set-up for computing the approximate posterior distributions is the same as that considered 
for the models illustrated in Figure \ref{fig:SL-HR-ET} of Section \ref{sec:simu}.
The summaries of results obtained from the posterior distribution are displayed in Figure \ref{fig:dataAnalysis}.
The first row reports the point-wise posterior means (red line) and $95\%$ credibility bands (in grey) of the angular density (left panel) and the Pickands dependence function (right panel).
The results regarding the Pickands dependence function suggest that the dependence structure is symmetric.
Results about both the Pickands dependence function and the angular density suggest a mild dependence, with posterior median values of the point masses $p_0$ and $p_1$ equal to $0.149$ and $0.093$, respectively.
The bottom-left panel of Figure \ref{fig:dataAnalysis} displays the prior (green dots) and posterior distributions (red dots) for the polynomial degree $k$, with median value $7$. The bottom-right panel displays the
predicted probabilities of joint exceedances, given by \eqref{lab:predictive}, 
for combinations of values ranging between $10$ and $100$. These results highlight that the probability of joint exceedances is also symmetric and hence the two variables can be considered exchangeable. However, as we have previously discussed,
the marginal distribution of monthly-maxima of GBP/USD log-returns is different from that of GBP/USD.
Therefore, bringing this small case study to a close,
we compute both conditional probabilities when the conditioning variable exceeds its 99$\%$ percentile, i.e. $\prob(\mbox{GBP/JPY} > q_{1} \, |\,   \mbox{GBP/USD} > q_{1})$ and $\prob(\mbox{GBP/USD} > q_{2} \, |\,   \mbox{GBP/JPY} > q_{2})$.
To do so, we proceed as follows. We calculate $q_{1}$ and $q_{2}$ as the $99\%$ percentiles of the marginal GEV distributions of log-returns of exchange rates GBP/USD and GBP/JPY, respectively, using the estimated marginal parameters.
These are equal to $q_{1} = 0.0162$ and $q_{2} = 0.0221$.
We transform the thresholds in order to represent them in unit-Fr\'{e}chet scale by
$$
y^{*}_{i,j} = \Bigg\{1 + \xi_{i} \bigg( \frac{q_{j} - \mu_{i}}{\sigma_{i}} \bigg) \Bigg\}_+^{(1/\xi_{i})}, \qquad i,j=1,2.
$$
Now, for $q_1$ we obtain the thresholds $y^{*}_{2,1} = 14.12$ and $y^{*}_{1,1} = 57.25$ and the
joint predictive probability \eqref{lab:predictive} is equal to 0.0050. 
Therefore, we obtain the final result
$\prob(\mbox{GBP/JPY} > q_{1} \, |\,   \mbox{GBP/USD} > q_{1}, )\approx\prob(Y_{2} > y^{*}_{2,1} \, |\,   Y_{1} > y^{*}_{1,1}) = 0.2880$.
Similarly, for $q_2$ we obtain the thresholds $y^{*}_{1,2} = 450.23$ and $y^{*}_{2,2} = 52.32$ and the
joint predictive probability \eqref{lab:predictive} is equal to 0.0007.
Therefore, we obtain the final result
$\prob(\mbox{GBP/USD} > q_{2} \, |\,   \mbox{GBP/JPY} > q_{2})\approx\prob(Y_{1} > y^{*}_{1,2} \, |\,   Y_{2} > y^{*}_{2,2}) = 0.0386$.

In conclusion, in contrast to the case of the joint exceedances, since the GBP/JPY tends to assume larger values than GBP/USD, then the conditional probability of the log-returns of GBP/USD given elevated values of log-returns of GBP/JPY is quite high.

%
\section*{Acknowledgements}\label{sec:ackn}
%

We thank a referee, an associate editor and the editor for useful suggestions and comments.

\newpage

\appendix

%
\section*{Appendix: Proofs}\label{sec:appA}
%

%
\begin{proof}[Proof of Proposition \eqref{prop:angular_dist}]

Using the identities in \eqref{eq:basis_conrners} we have that $H_{k-1}(0)=p_0$, $H_{k-1}(1)=1-p_1$, where 
the former is the atom in $\{0\}$ and $H_{k-1}(1)=\sup_{w\in[0,1)} H_{k-1}(w)$. As a result $H_{k-1}([0,1])-H_{k-1}(1)=p_1$, which is
the atom in $\{1\}$.

Second, for any $w_1\leq w_2\in [0,1)$ we have
$$
H_{k-1}([0,w_2])- H_{k-1}([0,w_1])=\sum_{j=0}^{k-1} (\eta_{j+1}-\eta_j) \int_{w_1}^{w_2} \betaf(v|j+1,k-j) \diff v \geq 0,
$$
where the inequality holds because by (R1) we have that $\eta_{j+1}-\eta_j\geq 0$, for $j=0,\ldots,k-1$ and therefore $H_{k-1}([0,w_1])\leq H_{k-1}([0,w_2])$.

Third, note that
\begin{equation}\label{eq:bpoly_mean_const1}
\int_0^1w\,h_{k-1}(w) \, \diff w + p_1 = p_1+ \frac{1}{{k}}\sum_{j=0}^{k-2} (\eta_{j+1}-\eta_{j})(j+1),
\end{equation}
and
\begin{equation}\label{eq:bpoly_mean_const2}
p_0+\int_0^1(1-w)\,h_{k-1}(w) \, \diff w = p_0 + \frac{1}{{k}}\sum_{j=0}^{k-2} (\eta_{j+1}-\eta_{j})(k-j-1).
\end{equation}
Equating \eqref{eq:bpoly_mean_const1} and \eqref{eq:bpoly_mean_const2} to $1/2$ we attain the
condition in (R2). Then $H_{k-1}$ satisfies the mean constraint (C1) by applying (R2) to its 
coefficients.
\end{proof}
\begin{proof}[Proof of Proposition~\ref{prop:equivalence_AH}]

By \eqref{eq:relation_pick_ang} we have that $H([0,w])=(A'(w)+1)/2$ for $w\in[0,1)$. Applying such a relationship between
$H_{k-1}$ and $A_k$ we attain

\begin{eqnarray*}\label{eq:rel_ang_pick}
H_{k-1}([0,w]) &=& \frac{1}{2}\left\{k\sum_{j=0}^{k-1}(\beta_{j+1}-\beta_{j})b_j(w;k-1) +1\right\}\\
&=& \sum_{j=0}^{k-1}\frac{1}{2}\left\{k(\beta_{j+1}-\beta_{j})+1\right\} b_j(w;k-1)\\
&=& \sum_{j=0}^{k-1}\eta_j b_j(w;k-1),
\end{eqnarray*}
where we have used the identity $\sum_{j\leq k-1} b_j(w;k-1)=1$. From the above formula
the result in \eqref{eq:etas} follows. On the other hand we have
\begin{eqnarray*}\label{eq:rel_pick_ang}
A'_k(t) &=& 2H_{k-1}([0,t]) -1\\
k\sum_{j=0}^{k-1}(\beta_{j+1}-\beta_{j})b_j(t;k-1) &=& \sum_{j=0}^{k-1} (2\eta_j-1 ) b_j(t;k-1)
\end{eqnarray*}
where the last identity holds if and only if $k(\beta_{j+1}-\beta_{j})=2\eta_j-1$ for all $j=0,\ldots k-1$.
Resolving for $\beta_{j+1}$ we attain the formula $\beta_{j+1}=\beta_j+(2\eta_j-1)/k$. From this
we get $\beta_1=(2\eta_0+k-1)/k$, for $j=0$, since $\beta_0=1$. Applying it recursively we get
$\beta_2=(2(\eta_0+\eta_1)+k-2)/k$, for $j=1$. Repeating this reasoning for $j=2,3,\ldots$ we attain
the general recursive formula in \eqref{eq:betas}. Thus, statement $i)$ is shown.

Consider $A_k$ in \eqref{eq:bpoly_picka} and assume it fulfills (R3)-(R5).
Then, we must check that $H_{k-1}$ in \eqref{eq:bpoly_angdist} with coefficients given by \eqref{eq:etas} 
fulfill (R1) and (R2). 

By \eqref{eq:etas} we have $\eta_0=k(\beta_1-1+1/k)/2$ and $\eta_{k-1}=k(1-\beta_{k-1}+1/k)/2$ for $j=0$ and $j=k-1$.
By (R4) we have therefore that $\eta_0=p_0$ and $\eta_{k-1}=1-p_1$ .
Next it needs to be shown that $\eta_j\leq \eta_{j+1}$ for $j=0,\ldots,k-2$. By \eqref{eq:etas} this inequality is equal to
$k(\beta_{j+1}-\beta_{j}+1/k)/2 \leq k(\beta_{j+2}-\beta_{j+1}+1/k)/2$ for $j=0,\ldots,k-2$. This holds if and only if
$\beta_{j+2}-2\beta_{j+1}+\beta_j\geq0$ and this is true by (R5). Thus $H_{k-1}$ fulfills (R1). 

It remains to show that $\eta_0+\cdots+\eta_{k-1}=k/2$. By \eqref{eq:etas} with a few steps we attain
\begin{eqnarray*}\label{eq:mean_con_verify}
\frac{k}{2}+p_1-1-p_0&=&\sum_{j=1}^{k-2}\eta_j\\
&=& \sum_{j=1}^{k-2}\left( \frac{k(\beta_{j+1}-\beta_{j} +1/k)}{2}\right)\\
&=&\frac{k}{2} -1 +\frac{k}{2}\sum_{j=1}^{k-2} (\beta_{j+1} - \beta_{j} ) 
\end{eqnarray*}
and from the last identity we obtain 
$
2(p_1-p_0)=\beta_{k-1}-\beta_1.  
$
By (R4) it is straightforward to check that the last equation holds. Therefore $H_{k-1}$ fulfills also (R2) and it is the distribution
of a valid angular measure.

Now, consider $H_{k-1}$ in \eqref{eq:bpoly_angdist} and assume it fulfills (R1) and (R2).
Then, we must check that $A_k$ in \eqref{eq:bpoly_picka} with coefficients given by \eqref{eq:betas}
fulfill (R3)-(R5).

Applying \eqref{eq:betas} with $j=k-1$ we have that $\beta_k=1$ and this is attained using the condition (R2). Next, it needs to be shown
that $\beta_{j+1}\leq 1$ for any $j=0,\ldots,k-1$. 
Applying \eqref{eq:betas} to check that such inequalities hold is equivalent to checking that $\sum_{i\leq j}\eta_i \leq (j+1)/2$ for any $j=0,\ldots,k-1$.
Thus, when $j=0$ we have $\eta_0\leq 1/2$ and this holds since that $\eta_0=p_0$ by (R1) and $p_0\in [0,1/2]$ by Assumption \ref{ass:measure_cond}. 
For any $j=1,\ldots,k-2$ suppose on the contrary that $(\eta_0+\cdots+\eta_j)>(j+1)/k$. 
From this and taking into account (R1) and that $p_1\in[0,1/2]$  by Assumption \ref{ass:measure_cond}, it follows the contradiction that (R2) is not valid. 
As a consequence the opposite inequalities hold.
Since $\beta_0=1$ by definition, then $A_k$ fulfills (R3).

By \eqref{eq:betas}, for $j=0$ and $j=k-2$,
we derive with some manipulations $\beta_1=(2p_0 + k -1)/k$ and $\beta_{k-1}=(2p_1+k-1)/k$. 
These results are attained by using (R1) and (R2), respectively. Therefore $A_k$ fulfills (R4).

It remains to show that $\beta_{j+2}-2\beta_{j+1}+\beta_j\geq0$ for all $j=0,\ldots k-2$. Applying \eqref{eq:betas} and with some manipulations
we have 
\begin{eqnarray*}\label{eq:convexity_verify}
0&\leq&\beta_{j+2}-2\beta_{j+1}+\beta_j \\
&\leq&\frac{1}{k}\left(2\sum_{i=0}^{j+1}\eta_i+k-j-2\right)-
\frac{2}{k}\left(2\sum_{i=0}^{j}\eta_i+k-j-1\right)+
\frac{1}{k}\left(2\sum_{i=0}^{j-1}\eta_i+k-j\right)\\
&\leq&\frac{2}{k}(\eta_{j+1}-\eta_j)
\end{eqnarray*}
for $j=0,\ldots k-2$. The last inequality holds since $\eta_j\leq \eta_{j+1}$ for $j=0,\ldots k-2$ by (R1).
Therefore $A_k$ fulfills also (R5) and is a valid Pickands dependence function. Then the proof is concluded.
\end{proof}
\begin{proof}[Proof of Proposition~\ref{prop:approximation_H_A}]

The fact that $\spA_k$, $k=1,2,\ldots$ is nested in $\spA$ has been shown by Proposition 3.3 in \shortciteN{marcon+p+n+m15}.
Here we only need to show that $A_{k+1}(t)$ satisfies the conditions (R2), where 
$$
A_{k+1}(t)=\sum_{j=0}^{k+1}\beta^*_j \; b_j(t;k+1),\quad \beta^{*}_j=\left(\beta_j\frac{k+1-j}{k+1}+\beta_{j-1}\frac{j}{k+1}\right).
$$
Applying the above formula we have $\beta^*_1=(k\beta_1+\beta_0)/(k+1)$ and
$\beta^*_k=(\beta_k+k\beta_{k-1})/(k+1)$. Substituting with $\beta_0=\beta_{k}=1$,
$\beta_1=(2p_0+k-1)/k$ and $\beta_{k-1}=(2p_1+k-1)/k$, we obtain $\beta^*_1=(2p_0+k)/(k+1)$ and $\beta^*_k=(2p_1+k)/(k+1)$.
Therefore, the result is shown. We now show that also $\spH_k$, $k=1,2,\ldots$ is nested in $\spH$. Let
$$
H_{k}(w)=\sum_{j=0}^k\eta^{*}_j \; b_j(w;k), \quad \eta^{*}_j=\eta_j\frac{k-j}{k}+\eta_{j-1}\frac{j}{k}.
$$
We can verify that $\eta^{*}_j\leq\eta^{*}_{j+1}$, for $j=0,\ldots,{k-1}$. Using the definition of $\eta_j^{*}$ we obtain
\begin{eqnarray*}
-\frac{j}{k}\left(\eta_j-\eta_{j-1}\right)&\leq& \frac{k-j-1}{k} \left(\eta_{j+1}-\eta_{j}\right)
\end{eqnarray*}
and the left-hand and right-hand side of the above inequality is always negative and positive, respectively, by
(R1). Therefore, also $H_{k}$ satisfies condition (R1). Furthermore,  we have
\begin{eqnarray*}
\frac{k+1}{2}&=& \sum_{j=0}^k\eta^{*}_j\\
&=& \sum_{j=0}^k\left(\eta_j\frac{k-j}{k}+\eta_{j-1}\frac{j}{k}\right)\\
&=& \sum_{j=0}^{k-1}\eta_j +\frac{1}{2},
\end{eqnarray*}
where the last equation holds by (R2). As a consequence also $H_{k}$ satisfies conditions in (R2)
and hence $\spH_k$, $k=1,2,\ldots$, is nested in $\spH$.

Now, let 
$$ 
B_A(w;k)=\sum_{j=0}^k A\left(\frac{j}{k}\right)b_j(w;k),\quad k=1,2,\ldots,
$$
then, by Proposition 3.1 in \shortciteN{marcon+p+n+m15} we have
$$
\sup_{w\in[0,1]}| B_A(w;k) - A(w)|\leq \frac{1}{2\sqrt{k}}.
$$
Therefore, by Proposition 3.3 in \shortciteN{marcon+p+n+m15} the result in \eqref{eq:convergence_A} follows.
Next, consider $H_{k-1}$ as in \eqref{eq:bpoly_angdist}, where $H_{k-1}(w)=(A'_k(w)+1)/2$ for $w\in[0,1)$ and $A_k$ as 
in \eqref{eq:bpoly_picka}, satisfying (R3)-(R5). Then, $H_{k-1}\in \spH$ by Proposition \ref{prop:equivalence_AH} and $\spH_k$ is nested
in $\spH$ as has been shown above. Furthermore,
let $\tilde{B}_H(w;k-1)=(B'_A(w;k)+1)/2$ for $w\in\spa$, then
$$
| \tilde{B}_H(w;k-1) - H(w)|=|B'_A(w;k) - A'(w)|, \quad w\in \spa.
$$
As a consequence the result \eqref{eq:convergence_A} implies 
that also result \eqref{eq:convergence_H} holds, by the uniform convergence of the first derivative of
convex functions (see Theorem 25.7 in \citeNP{rockafellar2015}).
\end{proof}
\bibliographystyle{chicago}
\bibliography{bernpoly}

\begin{thebibliography}{}

\bibitem[\protect\citeauthoryear{Antoniano-Villalobos and
  Walker}{Antoniano-Villalobos and Walker}{2013}]{antoniano2013}
Antoniano-Villalobos, I. and S.~G. Walker (2013).
\newblock Bayesian nonparametric inference for the power likelihood.
\newblock {\em Journal of Computational and Graphical Statistics\/}~{\em
  22\/}(4), 801--813.

\bibitem[\protect\citeauthoryear{Beirlant, Goegebeur, Segers, and
  Teugels}{Beirlant et~al.}{2004}]{beirlant+g+s+t04}
Beirlant, J., Y.~Goegebeur, J.~Segers, and J.~Teugels (2004).
\newblock {\em Statistics of Extremes: Theory and Applications}.
\newblock John Wiley \& Sons Ltd.,Chichester.

\bibitem[\protect\citeauthoryear{Beranger and Padoan}{Beranger and
  Padoan}{2015}]{boris+p2015}
Beranger, B. and S.~A. Padoan (2015).
\newblock Extreme dependence models.
\newblock In D.~Dey and J.~Yan (Eds.), {\em Extreme Value Modeling and Risk
  Analysis: Methods and Applications}. Chapman and Hall/CRC.

\bibitem[\protect\citeauthoryear{Berghaus, B{\"u}cher, and Dette}{Berghaus
  et~al.}{2013}]{berghaus2013}
Berghaus, B., A.~B{\"u}cher, and H.~Dette (2013).
\newblock Minimum distance estimators of the {Pickands} dependence function and
  related tests of multivariate extreme-value dependence.
\newblock {\em Journal de la Soci{\'e}t{\'e} Fran{\c{c}}aise de
  Statistique\/}~{\em 154\/}(1), 116--137.

\bibitem[\protect\citeauthoryear{Boldi and Davison}{Boldi and
  Davison}{2007}]{boldi+d07}
Boldi, M.~O. and A.~C. Davison (2007).
\newblock A mixture model for multivariate extremes.
\newblock {\em Journal of the Royal Statistical Society, Series B\/}~{\em
  69\/}(2), 217--229.

\bibitem[\protect\citeauthoryear{B{\"u}cher, Dette, and Volgushev}{B{\"u}cher
  et~al.}{2011}]{bucher2011}
B{\"u}cher, A., H.~Dette, and S.~Volgushev (2011).
\newblock New estimators of the {Pickands} dependence function and a test for
  extreme-value dependence.
\newblock {\em The Annals of Statistics\/}~{\em 39\/}(4), 1963--2006.

\bibitem[\protect\citeauthoryear{Cap\'{e}ra\`{a}, Foug\`{e}res, and
  Genest}{Cap\'{e}ra\`{a} et~al.}{1997}]{cap+f+g97}
Cap\'{e}ra\`{a}, P., A.-L. Foug\`{e}res, and C.~Genest (1997).
\newblock A nonparametric estimation procedure for bivariate extreme value
  copulas.
\newblock {\em Biometrika\/}~{\em 84}, 567--577.

\bibitem[\protect\citeauthoryear{Coles}{Coles}{2001}]{coles01}
Coles, S.~G. (2001).
\newblock {\em An Introduction to Statistical Modelling of Extreme Values}.
\newblock Springer, London.

\bibitem[\protect\citeauthoryear{de~Haan and Ferreira}{de~Haan and
  Ferreira}{2006}]{dehaan+f06}
de~Haan, L. and A.~Ferreira (2006).
\newblock {\em Extreme Value Theory: An Introduction}.
\newblock Springer.

\bibitem[\protect\citeauthoryear{Einmahl, Krajina, and Segers}{Einmahl
  et~al.}{2008}]{einmahl2008}
Einmahl, J., A.~Krajina, and J.~Segers (2008).
\newblock A method of moments estimator of tail dependence. bernoulli 14
  1003--1026.
\newblock {\em Mathematical Reviews (MathSciNet): MR2543584 Digital Object
  Identifier: doi\/}~{\em 10}.

\bibitem[\protect\citeauthoryear{Engel and West}{Engel and
  West}{2005}]{engel2005}
Engel, C. and K.~D. West (2005).
\newblock Exchange rates and fundamentals.
\newblock {\em Journal of Political Economy\/}~{\em 113\/}(3), 485--517.

\bibitem[\protect\citeauthoryear{Falk, H\"{u}sler, and Reiss}{Falk
  et~al.}{2010}]{falk+h+r10}
Falk, M., J.~H\"{u}sler, and R.~D. Reiss (2010).
\newblock {\em Laws of Small Numbers: Extremes and Rare Events\/} (Third ed.).
\newblock Birkh\"{a}user Boston.

\bibitem[\protect\citeauthoryear{Fils-Villetard, Guillou, and
  Segers}{Fils-Villetard et~al.}{2008}]{fil+g+s08}
Fils-Villetard, A., A.~Guillou, and J.~Segers (2008).
\newblock Projection estimators of {Pickands} dependence functions.
\newblock {\em The Canadian Journal of Statistics\/}~{\em 36\/}(3), 369--382.

\bibitem[\protect\citeauthoryear{Genest and Segers}{Genest and
  Segers}{2009}]{genest2009rank}
Genest, C. and J.~Segers (2009).
\newblock Rank-based inference for bivariate extreme-value copulas.
\newblock {\em The Annals of Statistics\/}~{\em 37\/}(5B), 2990--3022.

\bibitem[\protect\citeauthoryear{Ghosal}{Ghosal}{2001}]{ghosal2001}
Ghosal, S. (2001).
\newblock Convergence rates for density estimation with bernstein polynomials.
\newblock {\em The Annals of Statistics\/}~{\em 29\/}(5), 1264--1280.

\bibitem[\protect\citeauthoryear{Godsill}{Godsill}{2001}]{godsill2001}
Godsill, S.~J. (2001).
\newblock On the relationship between markov chain monte carlo methods for
  model uncertainty.
\newblock {\em Journal of Computational and Graphical Statistics\/}~{\em
  10\/}(2), 230--248.

\bibitem[\protect\citeauthoryear{Guillotte and Perron}{Guillotte and
  Perron}{2008}]{guillotte2008}
Guillotte, S. and F.~Perron (2008).
\newblock A bayesian estimator for the dependence function of a bivariate
  extreme-value distribution.
\newblock {\em Canadian Journal of Statistics\/}~{\em 36\/}(3), 383--396.

\bibitem[\protect\citeauthoryear{H\"{u}sler and Reiss}{H\"{u}sler and
  Reiss}{1989}]{husler+r89}
H\"{u}sler, J. and R.~Reiss (1989).
\newblock Maxima of normal random vectors: between independence and complete
  dependence.
\newblock {\em Statistics and Probability Letters\/}~{\em 7}, 283--286.

\bibitem[\protect\citeauthoryear{Kl{\"u}ppelberg, Kuhn, Peng,
  et~al.}{Kl{\"u}ppelberg et~al.}{2007}]{kluppelberg2007}
Kl{\"u}ppelberg, C., G.~Kuhn, L.~Peng, et~al. (2007).
\newblock Estimating the tail dependence function of an elliptical
  distribution.
\newblock {\em Bernoulli\/}~{\em 13\/}(1), 229--251.

\bibitem[\protect\citeauthoryear{Krajina}{Krajina}{2012}]{krajina2012}
Krajina, A. (2012).
\newblock A method of moments estimator of tail dependence in meta-elliptical
  models.
\newblock {\em Journal of Statistical Planning and Inference\/}~{\em 142\/}(7),
  1811--1823.

\bibitem[\protect\citeauthoryear{Lorentz}{Lorentz}{1986}]{lorentz53}
Lorentz, G.~G. (1986).
\newblock {\em Bernstein Polynominals\/} (Second ed.).
\newblock Chelsea Publishing Company, New York.

\bibitem[\protect\citeauthoryear{Madura}{Madura}{2014}]{madura2014}
Madura, J. (2014).
\newblock {\em Financial markets and institutions}.
\newblock Cengage learning.

\bibitem[\protect\citeauthoryear{Marcon, Padoan, Naveau, and Muliere}{Marcon
  et~al.}{2015}]{marcon+p+n+m15}
Marcon, G., S.~A. Padoan, P.~Naveau, and P.~Muliere (2015).
\newblock Nonparametric estimation of the pickands dependence function using
  bernstein polynomials.
\newblock {\em Journal of Statistical Planning and Inference, Under
  revision\/}.

\bibitem[\protect\citeauthoryear{Meese and Rogoff}{Meese and
  Rogoff}{1983}]{meese1983}
Meese, R.~A. and K.~Rogoff (1983).
\newblock Empirical exchange rate models of the seventies: Do they fit out of
  sample?
\newblock {\em Journal of international economics\/}~{\em 14\/}(1), 3--24.

\bibitem[\protect\citeauthoryear{Nikoloulopoulos, Joe, and Li}{Nikoloulopoulos
  et~al.}{2009}]{nikoloulopoulos2009}
Nikoloulopoulos, A.~K., H.~Joe, and H.~Li (2009).
\newblock Extreme value properties of multivariate t copulas.
\newblock {\em Extremes\/}~{\em 12\/}(2), 129--148.

\bibitem[\protect\citeauthoryear{Petrone}{Petrone}{1999a}]{petrone1999bayesian}
Petrone, S. (1999a).
\newblock Bayesian density estimation using bernstein polynomials.
\newblock {\em Canadian Journal of Statistics\/}~{\em 27\/}(1), 105--126.

\bibitem[\protect\citeauthoryear{Petrone}{Petrone}{1999b}]{petrone99}
Petrone, S. (1999b).
\newblock Random {B}ernstein polynomials.
\newblock {\em Scandinavian Journal of Statistics\/}~{\em 26}, 373--393.

\bibitem[\protect\citeauthoryear{Petrone and Wasserman}{Petrone and
  Wasserman}{2002}]{petrone2002}
Petrone, S. and L.~Wasserman (2002).
\newblock Consistency of bernstein polynomial posteriors.
\newblock {\em Journal of the Royal Statistical Society: Series B (Statistical
  Methodology)\/}~{\em 64\/}(1), 79--100.

\bibitem[\protect\citeauthoryear{Pickands}{Pickands}{1981}]{pickands81}
Pickands, III, J. (1981).
\newblock Multivariate extreme value distributions.
\newblock In {\em Proceedings of the 43rd session of the {I}nternational
  {S}tatistical {I}nstitute, {V}ol.\ 2 ({B}uenos {A}ires, 1981)}, Volume~49,
  pp.\  859--878, 894--902.
\newblock With a discussion.

\bibitem[\protect\citeauthoryear{Rockafellar}{Rockafellar}{2015}]{rockafellar2015}
Rockafellar, R.~T. (2015).
\newblock {\em Convex analysis}.
\newblock Princeton university press.

\bibitem[\protect\citeauthoryear{Sabourin and Naveau}{Sabourin and
  Naveau}{2014}]{sabourin2014}
Sabourin, A. and P.~Naveau (2014).
\newblock Bayesian dirichlet mixture model for multivariate extremes: A
  re-parametrization.
\newblock {\em Computational Statistics \&amp; Data Analysis\/}~{\em 71},
  542--567.

\bibitem[\protect\citeauthoryear{Stephenson}{Stephenson}{2004}]{stephenson2004}
Stephenson, A. (2004).
\newblock A user's guide to the `evd' package (version 2.1).
\newblock {\em Department of Statistics. Macquarie University. Australia\/}.

\bibitem[\protect\citeauthoryear{Tawn}{Tawn}{1990}]{tawn90}
Tawn, J.~A. (1990).
\newblock Modelling multivariate extreme value distributions.
\newblock {\em Biometrika\/}~{\em 77\/}(2), 245--253.

\end{thebibliography}
\end{document}